\documentclass[11pt]{article}
\usepackage{sectsty}

\usepackage{graphicx}

\usepackage{comment}

\usepackage{amsmath,amssymb}
\usepackage{bm}
\usepackage{enumerate}
\usepackage{amsthm}
\usepackage[all,2cell]{xy}
\usepackage{mathrsfs}
\usepackage{mathtools}

\mathtoolsset{showonlyrefs}

\usepackage{tikz-cd}
\usepackage{xcolor}
\DeclareUnicodeCharacter{0301}{*************************************}

\usepackage[colorlinks=true]{hyperref}
\hypersetup{urlcolor=blue, citecolor=red, linkcolor=blue}

\usetikzlibrary{cd}

\usepackage[a4paper, left=20mm, right=20mm, top=30mm, bottom=30mm]{geometry}

\usepackage{marginnote}

\usepackage{imakeidx}
\makeindex[intoc]

\usepackage[nottoc]{tocbibind}
\usepackage[colorinlistoftodos]{todonotes} 

\def\baselinestretch{1.1}
\theoremstyle{plain}
\newtheorem{theorem}{Theorem}[section]
\newtheorem{corollary}[theorem]{Corollary}
\newtheorem{proposition}[theorem]{Proposition}
\newtheorem{lemma}[theorem]{Lemma}

\theoremstyle{definition}
\newtheorem{definition}[theorem]{Definition}
\newtheorem{remark}[theorem]{Remark}
\newtheorem{example}[theorem]{Example}

\def\beq{\begin{equation}}
\def\eeq{\end{equation}}
\def\bea{\begin{eqnarray}}
\def\eea{\end{eqnarray}}
\def\beann{\begin{eqnarray*}}
\def\eeann{\end{eqnarray*}}
\def\ben{\begin{enumerate}}
\def\een{\end{enumerate}}
\def\bit{\begin{itemize}}
\def\eit{\end{itemize}}

\def\derpar#1#2{\frac{\partial{#1}}{\partial{#2}}}

\newcommand\restr[2]{{
  \left.\kern-\nulldelimiterspace 
  #1 
  \right|_{#2} 
}}

\newcommand{\R}{\mathbb{R}}

\renewcommand{\d}{\mathrm{d}}

\newcommand{\Cinfty}{\mathscr{C}^\infty}

\newcommand{\Tan}{\mathrm{T}}
\newcommand{\T}{\mathrm{T}}

\newcommand{\Lie}{\mathscr{L}}

\newcommand{\X}{\mathfrak{X}}

\newcommand{\bfX}{\mathbf{X}}

\newcommand{\pr}{\mathrm{pr}}

 \DeclareMathOperator{\Div}{Div}
\def\d{\mathrm{d}}

\newcommand{\parder}[2]{\frac{\partial #1}{\partial #2}}
\newcommand{\dparder}[2]{\dfrac{\partial #1}{\partial #2}}

\newcommand{\parderr}[3]{\frac{\partial^2 #1}{\partial #2\partial #3}}

\def\to{\rightarrow}
\def\tol{\longrightarrow}

\def\tkq{\oplus^k\T Q}

\def\r{\R}

\let\ds\displaystyle

 \def\bcr{\begin{color}{red}}
\def\bcb{\begin{color}{blue}}
\def\enc{\end{color}}

\title{\sf
Some contributions to $k$-contact Lagrangian field equations, symmetries and dissipation laws
}

\author{\sffamily 
Xavier Rivas$^{a,}$\thanks{xavier.rivas@unir.net\quad ORCID: 0000-0002-4175-5157}\, , 
Modesto Salgado$^{b,c,}$\thanks{modesto.salgado@usc.es\quad ORCID: 0000-0003-3982-1845}\,   
and 
Silvia Souto$^{b,c,}$\thanks{silviasouto.perez@usc.es\quad ORCID: 0000-0003-0755-1211}
\\[1ex]
\\[0.1ex]
\normalsize\itshape\sffamily 
$^a$Escuela Superior de Ingeniería y Tecnología,
\\
\normalsize\itshape\sffamily
Universidad Internacional de La Rioja,
Logroño, Spain. 
\\[0.1ex]
\normalsize\itshape\sffamily 
$^b$Departamento de Xeometría e Topoloxía, Universidade de Santiago de Compostela,
\\\normalsize\itshape\sffamily 
15782 Santiago de Compostela, Spain.
\\[0.1ex]
\normalsize\itshape\sffamily
$^c$Centro de Investigación e Tecnoloxía Matemática de Galicia (CITMAga),
\\ \normalsize\itshape\sffamily 
15782 Santiago de Compostela, Spain.
}

\begin{document}

\date{\today}

\maketitle

\begin{abstract}

It is well known that $k$-contact geometry is a suitable framework to deal with non-conservative field theories. In this paper, we study some relations between solutions of the $k$-contact Euler--Lagrange equations, symmetries, dissipation laws and Newtonoid vector fields. We review the $k$-contact Euler--Lagrange equations written in terms of $k$-vector fields and sections and provide new results relating the solutions in both approaches. We also study different kind of symmetries depending on the structures they preserve: natural (preserving the Lagrangian function), dynamical (preserving the solutions), and $k$-contact (preserving the underlying geometric structures) symmetries. For some of these symmetries, we provide Noether-like theorems relating symmetries and dissipation laws. We also analyse the relation between $k$-contact symmetries and Newtonoid vector fields. Throughout the paper, we will use the damped vibrating string as our main illustrative example.

\end{abstract}


\noindent\textbf{MSC\,2020 codes:} 70S05, 70S10, 70G45, 53C15, 53D10, 35R01


\pagestyle{myheadings}
\markright{\small\sf 
{\it X. Rivas, M. Salgado \textit{\&} S. Souto} ---
$k$-contact Lagrangian equations, symmetries and dissipation laws}

\newpage

{\setcounter{tocdepth}{2}
\def\baselinestretch{1}
\small
\def\addvspace#1{\vskip 1pt}
\parskip 0pt plus 0.1mm
\tableofcontents
}


\section{Introduction}

Since the 1950s, geometric methods have been used to provide descriptions of mechanical systems and field theories with many applications in mathematics, physics, engineering, etc. Some of the most frequent geometric structures involved in geometric mechanics and field theory are symplectic, $k$-symplectic or multisymplectic manifolds (see \cite{Car1991,DeLeon2015,relations} and references therein). In general, all these geometric methods are applied to conservative systems, that is, without any dissipation or loss of energy, both in the Lagrangian and Hamiltonian sides.

In the last decade, the interest in the geometrization of systems with dissipation of energy has risen drastically. Contact geometry \cite{Ban2016,Gei2008,Kho2013} is the suitable geometric framework to describe many types of damping \cite{Bra2017a,Bra2017b,Car2019,DeLeo2019b,Gaset2020,Liu2018}. This formulation has proved to be very useful in thermodynamics \cite{Bra2018,Sim2020}, quantum mechanics \cite{Cia2018}, circuit theory \cite{Got2016}, Lie systems \cite{LR-2022} and control theory \cite{Ram2017} among others \cite{Bra2020,deLeon2022,DeLeo2021b,LLLR-2022,GG22,GG22b}. Also, there have been several generalizations of contact geometry in order to describe non-conservative field theory, namely the $k$-contact \cite{Gaset2021b,Gaset2021,Gra2021}, $k$-cocontact \cite{Riv-2022} and multicontact \cite{LGMRR-2022} formulations. The field equations obtained by means of these formalisms, called the {\it  $k$-contact Euler--Lagrange equations}, coincide with the ones obtained from the recently devised generalized Herglotz principle \cite{Gaset2022}.

The study of symmetries of dynamical systems is of great interest because it can provide new ways of finding conservation (or dissipation) laws. In addition, reduction procedures can be performed in order to simplify the description of a system whose group of symmetries is known. Since the seminal work by E. Noether \cite{Noether1971} (see also \cite{Kosmann-Schwarzbach2011,Neeman1999}), the relation between symmetries and conserved quantities has been a topic of great relevance in mathematical physics and dynamical systems. Since the dawn of geometric mechanics, many papers have been devoted to the study of symmetries and conserved quantities of Hamiltonian and Lagrangian systems \cite{Carinena1989,Djukic1975,marmo86,Prince1985,vanderSchaft1983}. Recently, this study has been performed for contact and cocontact systems \cite{LML,Gaset2020,GLR-2023}, where the notion of conserved quantity has to be replaced with the notion of dissipated quantity.

The aim of this paper is to deepen in the study of the symmetries of non-conservative autonomous field theories using the $k$-contact formalism. There exist many types of symmetries depending on the structures they preserve. In this work, we are focused in natural symmetries (symmetries of the Lagrangian function), dynamical symmetries (those preserving the solutions) and $k$-contact symmetries (those preserving the underlying geometric structures). Some of these symmetries allow us to obtain dissipation laws following the ideas of E. Noether.  Throughout the work, several examples are used, in particular the vibrating string with damping.

The structure of the paper is as follows. Section \ref{sec:geometric-elements} is devoted to review some basic concepts on $k$-vector fields, integral sections and \textsc{sopde}s which are fundamental tools in this paper.
Roughly speaking, a \textsc{sopde} is a $k$-vector field whose integral sections are first prolongations of maps defined on the base manifold. In addition, some geometric structures in the tangent bundle of $k^1$-velocities of a manifold are introduced using the theory of lifts of functions and vector fields. These structures are necessary to develop the geometrical $k$-contact equations.
In Section \ref{section-k-contact-Euler-Lagrange} we introduce the $k$-contact Euler--Lagrange equations, or Herglotz--Euler--Lagrange equations, and several examples of systems of this form are provided. 
Section \ref{gemetricformalism} is devoted to describe how to obtain these equations geometrically via Poincaré--Cartan forms. We present an example of the geometric $k$-contact Lagrangian equations.
In Section \ref{relations} we discuss the relation between solutions of the $k$-contact Euler--Lagrange  equations and the contact Lagrangian $k$-vector fields, provided by the geometrical $k$-contact equations.

In Section \ref{dlcs} the notion of dissipation law is introduced and some examples are provided. The characterization of these laws in Lemma \ref{lem:12} is a new tool for the study of dissipation law at the rest of the paper. 
Section \ref{section:symmetries} is devoted to present several types of symmetry, depending on the structure they preserve. Along this section, several Noether-like theorems relating symmetries and dissipation laws are provided.
Finally, we generalize the notion of Newtonoid vector field, see \cite{BBS}, from $k$-symplectic geometry to the $k$-contact case and we give some results relating Newtonoid vector fields to $k$-contact symmetries.

Throughout the paper, all the manifolds are real, second countable and of class $\Cinfty$, and the mappings are assumed to be smooth.
Einstein's notation for sums over crossed repeated indices is hereafter assumed.

\section{Preliminaries}\label{sec:geometric-elements}

The notion of $k$-vector field is of great interest in the geometric study of partial differential equations. 
In this section, we review the main concepts on $k$-vector fields and integral sections. In addition, we give some insights on the natural structures of the bundle of $k^1$-velocities: the Liouville vector field and the canonical $k$-tangent structure. These structures allow us to define the notion of second-order partial differential equation. For more information, we refer to \cite{DeLeon2015} and references therein.

\subsection{\texorpdfstring{$k$}--vector fields and integral sections}

Let $M$ be a smooth $n$-dimensional manifold. Consider the Whitney sum of $k$ copies of its tangent bundle: $\oplus^k\T M = \T M\oplus\overset{k}{\dotsb}\oplus\T M$. We have the natural projections
\begin{equation*}
    \tau^k_M\colon \oplus^k\T M\tol M\,, \qquad\tau^{k,\alpha}_M\colon\oplus^k\T M\tol \T M\,, \qquad \alpha = 1,\dotsc,k\,.
\end{equation*}
A \textbf{$k$-vector field} on a manifold $M$ is a section $ {\bf X}\colon M\to\oplus^k\T M$ of the natural projection $\tau^k_M$. We will denote by $\X^k(M)$ the set of all $k$-vector fields on $M$. Thus, a $k$-vector field ${\bf X}\in\X^k(M)$ can be understood as a family of $k$ vector fields $X_1,\dotsc,X_k\in\X(M)$, given by $X_\alpha = \tau^{k,\alpha}_M\circ{\bf X}$. With this in mind, we can denote ${\bf X} = (X_1,\dotsc,X_k)$. A smooth map $F:M \to N$ induces the map $\oplus^k\T F:\oplus^k\T M \to  \oplus^k\T N$ given by
\begin{equation}\label{tf}
    \oplus^k\T_m F({v_1}_m,\dotsc,{v_k}_m)=(\T_mF({v_1}_m),\dotsc,\T_mF({v_k}_m)) \,,
\end{equation}
where $m\in M$, ${v_1}_m,\dotsc,{v_k}_m\in \T_mM$, and $\T_mF:\T_mM \to \T_{F(m)}N$ is the tangent map of $F$.

Given a map $\psi\colon U\subset\R^k\to M$, we define its \textbf{first prolongation} to $\oplus^k\T M$ as the map $\psi^{(1)}\colon U\subset \R^k \to \oplus^k\T M$ given by
\begin{equation}\label{firstpro}
  \psi^{(1)}(t)= \oplus^k\T_t \psi\left(\ds\frac{\partial}{\partial t^1}\bigg\vert_t,\dotsc,\ds\frac{\partial}{\partial t^k}\bigg\vert_t\right) = \left(\T_t\psi \left(\ds\frac{\partial}{\partial t^1}\bigg\vert_t\right),\dotsc, \T_t\psi\left(\ds\frac{\partial}{\partial t^k}\bigg\vert_t\right)\right)\,,    
\end{equation}
where $t = (t^1,\dotsc,t^k)$ are the canonical coordinates of $\R^k$.

In the same way as one has integral curves of vector fields, the notion of integral section of a $k$-vector field is defined as follows. Consider a $k$-vector field ${\bf X} = (X_1,\dotsc,X_k)\in\X^k(M)$. An \textbf{integral section} of ${\bf X}$ is a map $\psi\colon U\subset\R^k\to M$ such that $ \psi^{(1)} = {\bf X}\circ\psi$, namely $\T_t\psi\left({\parder{}{t^\alpha}}\big\vert_{t}\right) = (X_\alpha\circ\psi)(t)$ for every $\alpha=1,\dotsc,k$. We say that a $k$-vector field ${\bf X}\in\X^k(M)$ is \textbf{integrable} if every point of $M$ is in the image of an integral section of ${\bf X}$.

Consider a $k$-vector field ${\bf X} = (X_\alpha)$ with local expression
$ X_\alpha = X_\alpha^i\derpar{}{x^i}\, $
in a coordinate system $(x^i)$ on $U\subset M$.
Then, $\psi\colon U\subset\R^k\to M$ is an integral section of ${\bf X}$ if, and only if, it is a solution of the system of partial differential equations
$$ \derpar{\psi^i}{t^\alpha}\Big\vert_t = X_\alpha^i(\psi(t))\,. $$

A $k$-vector field ${\bf X} = (X_1,\dotsc,X_k)$ on $M$ is integrable if, and only if, $[X_\alpha,X_\beta] = 0$ for every $\alpha,\beta$, which are precisely the necessary and sufficient conditions for the integrability of the above system of partial differential equations \cite{Lee2013}. Consider a diffeomorphism $\Phi: M \to M$ and a $k$-vector field $\mathbf{X} = (X_1,\dotsc,X_k)$ on $M$. If $\psi$ is an integral map of $\bfX$, then $\Phi\circ \psi$  is an integral section of $\Phi_*\mathbf{X}=(\Phi_*X_\alpha)$. In particular, if $\mathbf{X}$ is integrable
then so is $\Phi_*\mathbf{X}$.


\subsection{The tangent bundle of \texorpdfstring{$k^1$}--velocities}

Let $\tau_Q : \T Q \to Q$ be the tangent bundle of a smooth manifold $Q$.
The vector bundle $\oplus^k \T Q$ is called {\bf the tangent bundle of $k^1$-velocities of
$Q$} (see \cite{DeLeon2015}). 

If $(q^i)$ are local coordinates on $U \subseteq Q$, the natural coordinates $(q^i , v^i)$ on $\T U=\tau_Q^{-1}(U)$, namely
$$
 q^i(v_q)=q^i(q)\,,\qquad
  v^i(v_q)=v_q(q^i) \,,
  $$
induce natural coordinates $(q^i , v_\alpha^i)$, with $i = 1,\dotsc,n$ and $\alpha = 1,\dotsc,k$, on $(\tau^k_Q)^{-1}(U)$ which are given by
$$ q^i({v_1}_q,\ldots , {v_k}_q)=q^i(q)\,,\qquad
  v_\alpha^i({v_1}_q,\ldots , {v_k}_q)={v_\alpha}_q(q^i) \, .$$

 \paragraph{Lifts of functions.}
 
If $f$ is a differentiable function on $Q$, the vertical lift $f^V$ and the $\alpha$-lift $f^{(\alpha)}$ of $f$, are the functions on $\oplus^k \T Q$ given by
$$f^V(v_q)=(\tau_Q^k)^*f(v_q)=f(q)
\,, \qquad f^{(\alpha)}(v_q)=v_{\alpha_q}(f)\equiv v^i_\alpha(v_q)\ds\derpar{f}{q^i}\Big\vert_q\, , \quad \text{for every} \,\,\, \alpha=1,\ldots,k\,.$$

Since $(q^i)^V=q^i$ and $(q^i)^{(\alpha)}=v^i_\alpha$, we deduce that vector fields on $\oplus^k \T Q$ are characterized by its action on vertical and $\alpha$-lifts of functions.

\paragraph{Lifts of vector fields.} Given a vector field $X\in\X(Q)$, the vertical $\alpha$-lifts $X^{V_\alpha}$, and the complete lift $X^C$ are the vector fields on $\oplus^k \T Q$ given by
\begin{equation}\label{VClifts}
\begin{array}{lcl}  
X^{V_\alpha}(f^V)=0, & &X^{V_\alpha}(f^{(\beta)})=\delta^\alpha_\beta (X(f))^V ,
\\ \noalign{\medskip}
X^C(f^V)=(X(f))^V, & & X^C(f^{(\alpha)})=(X(f))^{(\alpha)}  \, ,
 \end{array}
\end{equation}
with $\alpha=1,\ldots,k .$

Taking adapted coordinates $(q^i, v_\alpha^i)$ on $\oplus^k\T Q$, if $X=X^i \ds\parder{}{q^i}$, then
\begin{equation}\label{lifts-vectors1} 
X^{V_\alpha} =X^i \derpar{}{v^i_\alpha} \,,
\qquad 
X^C=X^i \derpar{}{q^i}+ v^j_\alpha\derpar{X^i}{q^j}\derpar{}{v^i_\alpha} \,,
 \end{equation} 
and we have
\begin{equation}\label{lifts-vectors2} 
\left(\derpar{}{q^i}\right)^{V_\alpha} =   \derpar{}{v^i_\alpha} \,,
\qquad 
\left(\derpar{}{q^i}\right)^C=  \derpar{}{q^i}\,  .\end{equation}

As a consequence of \eqref{lifts-vectors1} and \eqref{lifts-vectors2} we deduce that the tensor fields of type $(1,1)$ on $\oplus^k \T Q$ are characterized by the action on these lifts of vector fields.

The \textbf{Liouville vector field} $\Delta\in\X(\oplus^k\T Q)$ is the vector field defined by $\Delta(f^V)=0$ and $\Delta (f^{(\alpha)})=f^{(\alpha)}$. In adapted coordinates, it reads
\begin{equation}\label{delta1}
    \Delta = v^i_\alpha \derpar{}{v^i_\alpha}\,.
\end{equation}
 
\begin{remark}\
\begin{enumerate}[{\rm (1)}]
\item If $X$ is a vector field on $Q$ with local one-parameter group $\phi_t: Q \to Q$, then $X^C$ is the infinitesimal generator of the flow $\T\phi_t : \tkq \to \tkq$.

\item $\Delta$ is the infinitesimal generator of the flow 
$\phi : \R \times \tkq \to \tkq \,,\, \phi(t,(v_{1_q},\ldots,v_{k_q})) = (e^t v_{1_q}, \ldots,e^t v_{k_q}) $.
\end{enumerate}\end{remark}

The \textbf{canonical $k$-tangent structure} is the family $(J^1,\ldots, J^k)$ of tensor fields
of type $(1,1)$ defined by
$J^\alpha(X^C)=X^{V_\alpha}$ and $J^\alpha(X^{V_\alpha} )= 0 \, .$
In local adapted coordinates it is written as
\begin{equation}\label{localJ}
J^\alpha=\displaystyle\frac {\displaystyle\partial}{\displaystyle\partial v^i_\alpha} \otimes \d q^i \,
.
\end{equation}

\subsection{Second-order partial differential equations}\label{kvf}

Since $\oplus^k \Tan Q\times\R^k\to\oplus^k \Tan Q$ is a trivial vector bundle, 
the canonical structures in $\oplus^k \Tan Q$ (the 
canonical $k$-tangent structure and the
Liouville vector field described above)
can be extended to $\oplus^k \Tan Q\times\R^k$ in a natural way,
and are denoted with the same notation ($J^\alpha$ and $\Delta$).
\begin{definition}\label{dfn:k-contact-holonomic-section}
    Consider a map $\phi_s\colon U \subset\R^k\to Q\times\R^k$ with
    $ \phi_s(t)= (\phi(t),s^\alpha(t))\,, $
    where $\phi\colon U \subset \R^k\to Q$. The \textbf{first prolongation} of $\phi_s$ to $\oplus^k\T Q \times\R^k$ is the map $\phi_s^{(1)}\colon U \subset \R^k\to \oplus^k\T Q \times\R^k$ given by
    $$ \phi_s^{(1)}(t)= (\phi^{(1)}(t),s^\alpha(t))\,, $$
    where
    $$\phi^{(1)}(t) = \left(\T_t\phi\left(\ds\frac{\partial}{\partial
t^1} \Big\vert_t\right), \ldots ,
\T_t\phi\left(\ds\frac{\partial}{\partial t^k}\Big\vert_t \right)
\right)$$
    is the first prolongation of $\phi$ to $\oplus^k\T Q$ defined in \eqref{firstpro}, and $(t^1,\ldots,t^k)$ are the canonical coordinates of $\r^k$.
\end{definition}

In local coordinates, if $\phi(t)=(\phi^i(t))$, then we have
\begin{equation}\label{phis1}
    \phi_s^{(1)}(t) = \left( \phi^i (t), \frac{\partial\phi^i}{\partial t^\alpha} \Big\vert_t,s^\alpha(t)\right).
\end{equation}
 
\begin{definition}\label{xijso}
    A $k$-vector field $\mathbf{\Gamma}=(\Gamma_1,\ldots,\Gamma_k)$ on $    \oplus^k\T Q \times\R^k$ is a \textbf{second-order partial differential equation} (or a {\sc sopde}) if, and only if, $J^{\alpha}(\Gamma_{\alpha})=\Delta$.
\end{definition}
In local adapted coordinates, a {\sc sopde} $\mathbf{\Gamma}=(\Gamma_1,\ldots,\Gamma_k) $ is given by
\begin{equation}\label{localsode1}
    \Gamma_\alpha = v^i_\alpha\parder{}{q^i} + \Gamma_{\alpha \beta}^i \parder{}{v^i_\beta} + \Gamma_\alpha^\beta \parder{}{s^\beta}\,,\qquad  \alpha,\beta = 1,\dotsc,k\,,
\end{equation}
where $\Gamma_{\alpha\beta}^i$ and $\Gamma_\alpha^\beta$ are smooth functions defined on domains of induced charts on $\oplus^k\T Q\times\R^k$.

If $\psi\colon U\subset \R^k \to \oplus^k\T Q\times\R^k$, locally given by
$\psi(t)=(\psi^i(t),\psi_\alpha^i(t),s^\alpha(t))$, is an integral section of a {\sc sopde} $\mathbf{\Gamma}=(\Gamma_1,\ldots,\Gamma_k)$, from Definition \ref{dfn:k-contact-holonomic-section} and equation \eqref{localsode1}, it follows that
\begin{equation}\label{solsopde}
    \parder{\psi^i}{t^\alpha}\Big\vert_t = \psi^i_\alpha(t)\,,\qquad \parder{\psi^i_\beta}{t^\alpha}\Big\vert_t = \Gamma_{\alpha \beta}^i(\psi(t))\,,\qquad \parder{s^\beta}{t^\alpha}\Big\vert_t = \Gamma_\alpha^\beta(\psi(t))\,.
\end{equation}
Then, we have
$$
    \psi(t) = (\psi^i(t),\psi_\alpha^i(t),s^\alpha(t)) = \phi_s^{(1)}(t)=   (\phi^{(1)}(t),s^\alpha(t))\,, \qquad \parder{s^\beta}{t^\alpha}\Big\vert_t = \Gamma_\alpha^\beta(\phi_s^{(1)})\,,
$$
where $\phi = \pr_Q\circ\psi\colon U\subset \R^k \overset{\psi}{\tol} \oplus^k\T Q\times\R^k \overset{\pr_Q}{\tol} Q$.
    
Thus we obtain the
following characterization for the integral maps of a {\sc sopde}.
\begin{proposition}
    Let $\mathbf{\Gamma}=(\Gamma_1,\ldots,\Gamma_k)$ be an integrable
    {\sc sopde} on $\oplus^k\T Q\times\R^k$. If $\psi : U \subset \r^k \to \tkq \times \r^k$ is an integral section of $\mathbf{\Gamma}$,
    then
    $$ \psi(t)=(\phi^{(1)}(t), s^\alpha(t))=\phi_s^{(1)}(t)\,, $$ 
    and $\phi_s(t)=(\phi^i(t),s^\alpha(t))$ is a solution of the following system of second-order partial
    differential equations
    \begin{equation}\label{nn1}
        \parder{^2 \phi^i}{t^\alpha\partial t^\beta}\Big\vert_t = \Gamma_{\alpha \beta}^i\left(\phi^j(t), \frac{\partial\phi^j}{\partial t^\gamma}\Big\vert_t,s^\gamma(t)\right)\,, \qquad \parder{s^\beta}{t^\alpha}\Big\vert_t = \Gamma_\alpha^\beta\left(\phi^j(t), \frac{\partial\phi^j}{\partial t^\gamma}\Big\vert_t,s^\gamma(t)\right)\,.
    \end{equation}
    Conversely, if $\phi_s\colon U\subset \R^k \to Q\times \R^k$ is a map satisfying the system \eqref{nn1}, then 
    $\phi_s^{(1)}$
    is an integral section of
    $\mathbf{\Gamma}=(\Gamma_1,\ldots,\Gamma_k)$.
\end{proposition}

  \begin{definition} A map  $\phi_s\colon U\subset \R^k \to Q\times \R^k$ satisfying \eqref{nn1} is called a \textbf{solution} of the {\sc sopde} vector field $\mathbf{\Gamma}$.
  \end{definition}

\section{\texorpdfstring{$k$}--contact Euler--Lagrange equations}\label{section-k-contact-Euler-Lagrange}

In this section we introduce the $k$-contact Euler--Lagrange equations, also called Herglotz--Euler--Lagrange equations, and several examples of systems of this form are provided.

The $k$-contact Euler--Lagrange equations for a Lagrangian function $L:\oplus^k\T Q\times\R^k\to\R$ read 

\begin{equation}\label{eq:k-contact-Euler--Lagrange-equations}
    \parder{}{t^\alpha} \left(\parder{L}{v^i_\alpha} \circ \phi_s^{(1)}\right) - \parder{L}{q^i} \circ \phi_s^{(1)} = \left(\parder{L}{s^\alpha} \parder{L}{v^i_\alpha}\right) \circ \phi_s^{(1)}\,, \qquad \parder{s^\alpha}{t^\alpha} = L \circ \phi_s^{(1)}\,,
\end{equation}
for a map $\phi_s\colon U \subset \R^k\to Q\times\R^k, \,\, \phi_s(t)= (\phi(t),s^\alpha(t))$, where $\phi\colon U \subset \R^k\to Q$. Equations \eqref{eq:k-contact-Euler--Lagrange-equations} can be written as
\begin{equation}\label{eq:eulerlagrange2} 
    \parderr{L}{v_\alpha^i}{v_\beta^j} \parderr{\phi^j}{t^\alpha}{t^\beta} + \parderr{L}{q^j}{v^i_\alpha} \parder{\phi^j}{t^\alpha} + \parderr{L}{s^\beta}{v^i_\alpha} \parder{s^\beta}{t^\alpha} - \parder{L}{q^i} = \parder{L}{s^\alpha}\parder{L}{v^i_\alpha}\,,\qquad \parder{s^\alpha}{t^\alpha} = L\,,
\end{equation}
where $\phi_s^{(1)}(t) = \left(\phi^j(t), \dfrac{\partial\phi^j}{\partial t^\alpha}\bigg\vert_t,s^\alpha(t)\right)$, with $t\in\R^k$, which represents a system of second-order partial differential equations on $Q \times \R^k$. It is important to point out that these equations can be obtained from two different variational principles \cite{Gaset2022} generalizing the usual Herglotz variational principle \cite{Her1930} and the variational principle for the $k$-symplectic Euler--Lagrange equations \cite[Sec.\,6.2]{DeLeon2015}. Throughout the paper, we will call solutions of equations \eqref{eq:k-contact-Euler--Lagrange-equations} to the maps $\phi_s$ or $\phi_s^{(1)}$ indistinctly, since we are only considering solutions which are first prolongations of maps $\phi_s:U\subset\R^k\to Q\times\R^k$.

\begin{remark}
In the case $k=1$, equations \eqref{eq:k-contact-Euler--Lagrange-equations} are Herglotz's equations (see \cite{Her1930})
$$
    \frac{\d}{\d t} \left(\parder{L}{v^i} \circ \gamma\right) - \parder{L}{q^i} \circ \gamma= \left(\parder{L}{s} \parder{L}{v^i}\right) \circ \gamma\,, \qquad \frac{\d s}{\d t} = L \circ \gamma\,,
$$
for a Lagrangian function $L:\T Q\times \R \longrightarrow \R$, 
with solution a curve  $\gamma(t)=(q^i(t),\dot q^i(t),s(t))$ on $\T Q\times \r$.
\end{remark}

\begin{remark}
    Note that when the Lagrangian function $L$ does not depend on the variables $s^\alpha$, for all $\alpha=1,\ldots,k$, then the above equations are the Euler--Lagrange field equations for an autonomous Lagrangian $L(q^i,v^i_\alpha)$ defined on $\oplus^k\T Q$.
\end{remark}

Let us see some examples of systems modelled by second-order partial differential equations which can be described by means of $k$-contact systems (see \cite{Gaset2021b}).

\begin{example}\label{exmpl:damped-vibrating-string} {\bf (The damped vibrating string).} 
    It is well known that a vibrating string can be described using the $k$-symplectic Lagrangian formalism. Consider the coordinates $(t,x)$ for the time and the space. Denote by $\phi$ the separation of a point in the string from its equilibrium point, and hence $v_t$ and $v_x$ will denote the derivative of $\phi$ with respect to the two independent variables. The Lagrangian function $L_\circ\colon \oplus^2\T\R\to\R$ for this system is given by
    \begin{equation}\label{eq:Lagrangian-vibrating-string}
        L_\circ(q,v_t,v_x) = \frac{1}{2}\rho v_t^2 - \frac{1}{2}\tau v_x^2\,,
    \end{equation}
    where $\rho$ is the linear mass density of the string and $\tau$ is the tension of the string. We assume that these quantities are constant. The Euler--Lagrange equation for this Lagrangian density is
    $$ \parder{^2\phi}{t^2} = c^2 \parder{^2\phi}{x^2}\,, $$
    where $c^2 = \dfrac{\tau}{\rho}$, which is the one-dimensional wave equation.
    
    In order to model a vibrating string with linear damping, we modify the Lagrangian function \eqref{eq:Lagrangian-vibrating-string} so that it becomes a $k$-contact Lagrangian function \cite{Gaset2021}. The new Lagrangian function $L$ is defined in the phase bundle $\oplus^2 \T \R\times \R^2$, endowed with coordinates $(q, v_t, v_x, s^t, s^x)$, as
    \begin{equation}\label{eq:Lagrangian-vibrating-string-dis}
    	L(q, v_t, v_x, s^t, s^x) = L_\circ - \gamma s^t = \frac{1}{2}\rho v_t^2 - \frac{1}{2}\tau v_x^2 - \gamma s^t\,.
    \end{equation}
    Consider a solution $\phi_s\colon\R^2\to \R\times\R^2$ with 
  $$\phi_s(t,x)= (\phi(t,x),s^t(t,x),s^x(t,x)) ,$$
  where $\phi\colon\R^2\to\R$. The first prolongation of $\phi_s$ is the map $\phi_s^{(1)}\colon\R^2\to \oplus^2\T\R \times\R^2$ given by
    \begin{equation}
      \phi_s^{(1)}(t,x) = \left( \phi (t,x),\parder{\phi}{t}\Big\vert_{(t,x)}, \parder{\phi}{x}\Big\vert_{(t,x)}, s^t(t,x), s^x(t,x)\right)\,.
    \end{equation}
    The $k$-contact Euler--Lagrange equations \eqref{eq:eulerlagrange2}, for the Lagrangian $L$, become
    \begin{equation}\label{eq:damped-vibrating-string-edp}
    	\frac{\partial^2 \phi}{\partial t^2} - c^2	\frac{\partial^2 \phi}{\partial x^2} + \gamma\parder{\phi}{t} = 0\,,\qquad \parder{s^1}{t} + \parder{s^2}{x} = L\circ\phi_s^{(1)}\,.
    \end{equation}
    The first equation corresponds to a vibrating string with damping.
\end{example}

\begin{example}\label{exmpl:two-coupled-vibrating-strings} {\bf (Two coupled vibrating strings with damping).}
    Consider a system of two coupled strings with damping. The configuration manifold of the system is $Q = \R^2$ equipped with coordinates $(q^1,q^2)$, where each coordinate represent the displacement of each string. The Lagrangian phase bundle of this system is $M = \oplus^2\T\R^2\times\R^2$ endowed with natural coordinates $(q^1,q^2,v_1^t,v_2^t,v_1^x,v_2^x,s^t,s^x)$. Consider the Lagrangian function $L\colon \oplus^2\T\R^2\times\R^2\to\R$ given by
    $$ L(q^1,q^2,v_t^1,v_t^2,v_x^1,v_x^2,s^t,s^x) = \frac{1}{2}\left( (v_t^1)^2 + (v_t^2)^2 - (v_x^1)^2 - (v_x^2)^2 \right) - C(z) - \gamma s^t\,, $$
    where $\gamma$ is a friction coefficient and $C$ is a function that represents a coupling of the two strings that depends only on $z = \sqrt{(q^1)^2 + (q^2)^2}$. 
    
    The maps $\phi^{(1)}_s = \left(\phi^1,\phi^2,\parder{\phi^1}{t},\parder{\phi^1}{x},\parder{\phi^2}{t},\parder{\phi^2}{x},s^t,s^x\right)$  solution to the $k$-contact Euler--Lagrange equations \eqref{eq:k-contact-Euler--Lagrange-equations} satisfy the system of partial differential equations
    \begin{gather}
        \parder{^2\phi^1}{t^2} - \parder{^2 \phi^1}{x^2} + \gamma\parder{\phi^1}{t} + C'(z)\frac{\phi^1}{z} = 0\,,\\
        \parder{^2\phi^2}{t^2} - \parder{^2 \phi^2}{x^2} + \gamma\parder{\phi^2}{t} + C'(z)\frac{\phi^2}{z} = 0\,,\\
        \parder{s^t}{t} + \parder{s^x}{x} = L\circ\phi_s\,.
    \end{gather}
    The first two equations correspond to two damped coupled strings with coupling function $C$. The Hamiltonian formulation of this system was studied in \cite{Gaset2021}.
\end{example}

\begin{example}\label{exmpl:telegrapher} {\bf (The telegrapher's equation).}
    The current and voltage on a uniform electrical transmission line is described by the so-called {\it telegrapher's equations} \cite[p.\,306]{Hay2018}, \cite[p.\,653]{Sal2015}:
    \begin{equation*}
        \begin{dcases}
            \parder{V}{x} = -L\parder{I}{t} - RI\,,\\
            \parder{I}{x} = -C\parder{V}{t} - GV\,,
        \end{dcases}
    \end{equation*}
    where $V(x,t)$ is the voltage, $I(x,t)$ is the current, $L$ is the inductance, $R$ is the resistance, $C$ is the capacitance, and $G$ is the conductance. This system can be uncoupled, obtaining the system of second-order partial differential equations
    \begin{equation*}
        \begin{dcases}
            \parder{^2V}{x^2} = LC\parder{^2V}{t^2} + (LG + RC)\parder{V}{t} + RGV\,,\\
            \parder{^2I}{x^2} = LC\parder{^2I}{t^2} + (LG + RC)\parder{I}{t} + RGI\,.\\
        \end{dcases}
    \end{equation*}
    Note that both equations in the system above are identical, and are also known as telegrapher's equation. Both of them can be written as
    \begin{equation}\label{eq:telegraph-equation}
        \square q + \gamma\parder{q}{t} + m^2 q = 0\,,
    \end{equation}
    where $$\gamma = \dfrac{LG + RC}{LC}\,,\qquad m^2 = \dfrac{RG}{LC}\,,\qquad \square = \dparder{^2}{t^2} - \dfrac{1}{LC}\dparder{^2}{x^2}\,.$$ 
    
    Equation \eqref{eq:telegraph-equation} is the $k$-contact Euler--Lagrange equation (see \cite{Gra2021}) for the Lagrangian function $L:\oplus^2\T\R\times\R^2\to\R$ given by
    $$ L(q, v_t, v_x, s^t, s^x) = \frac{1}{2}v_t^2 - \frac{1}{2LC}v_x^2 - \frac{1}{2}m^2q^2 - \gamma s^t\,. $$
\end{example}

\begin{example}{\bf (Laplace's equation with damping).}
    It is well known that Laplace's equation is
    $$ \Delta \phi := \parder{^2\phi}{x_1^2} + \dotsb + \parder{^2\phi}{x_n^2} = 0\,, $$
    where $\phi: \R^n \to{\R}$,
    and it can be understood as the Euler--Lagrange equation for a Lagrangian $L_\circ\colon\oplus^n\T\R\to\R$ given by
    $$ L_\circ(q,v_1,\dotsc,v_n) = \frac{1}{2}(v_1^2 + \dotsb + v_n^2)\,. $$
    Consider now the phase bundle $\oplus^n\T\R\times\R^n$ with canonical coordinates $(q, v_1,\dotsc,v_n, s^1,\dotsc,s^n)$ and the $n$-contact Lagrangian function $L\colon\oplus^n\T\R\times\R^n\to\R$ given by
    $$ L = L_\circ - \gamma_is^i\,, $$
    where $\gamma_i\in\R$ are constants. The $n$-contact Euler--Lagrange equations \eqref{eq:eulerlagrange2} for this Lagrangian $L$ become
    \begin{equation}\label{eq:damped-laplace-eq}
        \parder{^2\phi}{x_1^2} + \dotsb + \parder{^2\phi}{x_n^2} + \gamma_i\parder{\phi}{x_i} = 0\,,\quad \parder{s^1}{x_1} + \dotsb + \parder{s^n}{x_n} = L\,,
    \end{equation}
    where $\phi^{(1)}_s : \R^n \to \oplus^n\T\R\times\R^n$, is the first prolongation of the map $\phi$, given by $\phi^{(1)}_s=\left(\phi, \parder{\phi}{x_1},\ldots,\parder{\phi}{x_n},s^1,\ldots,s^n\right)$.
    
    Note that the first equation in \eqref{eq:damped-laplace-eq} is an elliptic partial differential equation that can be understood as a Laplace's equation with some additional first-order terms.
\end{example}


\section{Geometric \texorpdfstring{$k$}--contact Lagrangian equations}\label{gemetricformalism}

In this section we review some of the main aspects of the $k$-contact  Lagrangian formalism for non-conservative field theories, first introduced in \cite{Gaset2021}.

\begin{definition}\label{def:lagrangian-elements}
A \textbf{Lagrangian function} is a function $L\in \Cinfty(\oplus^k\T Q\times\R^k)$.
\begin{itemize}
\item The \textbf{Lagrangian energy} associated to $L$ is the function defined by $E_L = \Delta(L) - L\in \Cinfty(\oplus^k\T Q\times\R^k)$.
\item The \textbf{Poincaré--Cartan forms} associated to $L$ are $$ \theta^\alpha_L = \d L\circ J^\alpha\in\Omega^1(\oplus^k\T Q\times\R^k)\,. $$
\item We define the following one-forms associated to $L$
    $$ \eta^\alpha_L = \d s^\alpha - \theta^\alpha_L\in\Omega^1(\oplus^k\T Q\times\R^k)\,, $$
    called \textbf{contact Lagrangian one-forms}.
\end{itemize}
\end{definition}
    Note that the contact Lagrangian forms introduced above are not contact forms. However, in favourable cases, they define a $k$-contact structure on $\oplus^k\T Q\times\R^k$, motivating their name. In natural coordinates, the local expressions of the objects introduced in Definition \ref{def:lagrangian-elements} are
    \begin{equation}\label{energy}
    E_L = v^i_\alpha\parder{L}{v^i_\alpha} - L\,,
    \end{equation}
    \begin{equation}\label{eq:contact-forms}
        \eta^\alpha_L = \d s^\alpha - \ds\frac{\partial L}{\partial v^i_\alpha} \d q^i , \quad  \quad \d \eta^\alpha_L = \ds\frac{\partial^2 L}{\partial q^j \partial v^i_\alpha} \d q^i \wedge \d q^j + \ds\frac{\partial^2 L}{\partial v^j_\beta \partial v^i_\alpha} \d q^i \wedge \d v^j_\beta + \ds\frac{\partial^2 L}{\partial s^\beta \partial v^i_\alpha} \d q^i \wedge \d s^\beta .
    \end{equation}

\begin{definition}
     A Lagrangian function $L\colon\oplus^k\T Q\times\R^k\to\R$ is said to be \textbf{regular} if the Hessian of the Lagrangian function $L$ with respect to the fibre coordinates, namely
    \begin{equation}\label{eq:hessian}
        g^{\alpha\beta}_{ij} = \parderr{L}{v^i_\alpha}{v^j_\beta}\,,
    \end{equation}
    has maximal rank $nk$ on $\oplus^k\T Q\times\R^k$. Otherwise, the Lagrangian function is \textbf{singular}.
\end{definition}

\begin{remark}
    A Lagrangian function $L:\oplus^k\T Q\times\R^k\to\R$ is regular if, and only if, the family of $k$ differential one-forms $\eta^1_L,\ldots,\eta^k_L$ defines a $k$-contact structure on $\oplus^k\T Q\times\R^k$. In this case, $(\oplus^k\T Q\times\R^k,\eta^\alpha_L)$ becomes a \textbf{$k$-contact manifold} (see \cite{Gaset2021b}).
\end{remark}
 
 Using the forms $\eta^\alpha_L$ we can write equations 
\eqref{eq:k-contact-Euler--Lagrange-equations} as follows
	\begin{equation}\label{eq:k-contact-Euler-Lagrange-section}
		\begin{dcases}
			i_{\displaystyle\T\phi_s^{(1)}\Big(\parder{}{t^\alpha}\Big)}\d\eta^\alpha_L = \left( \d E_L  + \dparder{L}{s^\alpha} \eta^\alpha_L \right)\circ\phi_s^{(1)}\,,\\
            i_{\displaystyle\T\phi_s^{(1)}\Big(\parder{}{t^\alpha}\Big)}\eta^\alpha_L = -E_L\circ\phi_s^{(1)}\,,
		\end{dcases}
	\end{equation}
 where 
$$
\T_t\phi_s^{(1)}\left(\parder{}{t^\alpha}\Big\vert_t\right)=
\parder{\phi^i}{t^\alpha}\Big\vert_t \,
 \parder{}{q^i}\Big\vert_{\phi_s^{(1)}(t)}
+
\parderr{\phi^i}{t^\alpha}{t^\beta}\Big\vert_t \,
 \parder{}{v^i_\beta}\Big\vert_{\phi_s^{(1)}(t)}
+
\parder{s^\beta}{t^\alpha}\Big\vert_t \,
 \parder{}{s^\beta}\Big\vert_{\phi_s^{(1)}(t)} \, .
$$

\begin{definition}
    The \textbf{geometric $k$-contact Lagrangian equations} for a $k$-vector field $\bfX=(X_1,\dotsc,X_k)$ on $\oplus^k\T Q\times\R^k$ are
    \begin{equation}\label{eq:k-contact-Lagrangian-fields}
        \begin{dcases}
            i_{X_\alpha} \d\eta_L^\alpha = \d E_L + \dparder{L}{s^\alpha}  \eta_L^\alpha \,,\\
            i_{X_\alpha}\eta_L^\alpha = -E_L \,.
        \end{dcases}
    \end{equation}
    We will denote by $\X^k_L(\oplus^k\T Q\times\R^k)$ the set of \textbf{$k$-contact Lagrangian $k$-vector fields}, namely the $k$-vector fields $\bfX = (X_1,\dotsc,X_k)$ on $\oplus^k\T Q\times\R^k$, which are solutions to the equations \eqref{eq:k-contact-Lagrangian-fields}.

    \end{definition}
 
For a $k$-vector field ${\bfX}=(X_1,\dotsc,X_k)\in\X^k_L(\oplus^k\T Q\times\R^k)$ with local expression
$$ X_\alpha = X_\alpha^i\parder{}{ q^i} + X_{\alpha\beta}^i\parder{}{v^i_\beta} + X_\alpha^\beta\parder{}{s^\beta}\,, $$
equations \eqref{eq:k-contact-Lagrangian-fields} read
 \begin{align}
    0 &= \left(X_\alpha^j - v^j_\alpha\right) \parderr{L}{s^\beta}{v^j_\alpha} \,,
    \label{primeq}\\
    0 &= \left(X_\alpha^j - v^j_\alpha\right) \parderr{L}{v^i_\beta}{v^j_\alpha}\,, \label{A-E-L-eqs2}\\
    0 &= \left(X_\alpha^j - v^j_\alpha\right) \parderr{L}{q^i}{v^j_\alpha} + \parder{L}{q^i} - X_\alpha^\beta\, \parderr{L}{s^\beta}{v^i_\alpha} - X_\alpha^j \, \parderr{L}{q^j}{v^i_\alpha} - X_{\alpha\beta}^j\, \parderr{L}{v^j_\beta}{v^i_\alpha} + \parder{L}{s^\alpha} \parder{L}{v^i_\alpha}\,,\label{A-E-L-eqs3}\\
    0 &= L + \left(X_\alpha^j - v^j_\alpha\right) \parder{L}{v^i_\alpha} - X_\alpha^\alpha\,. \label{A-E-L-eqs4}
\end{align}
If $L$ is a regular Lagrangian, equations \eqref{A-E-L-eqs2} lead to $X_\alpha^j = v^j_\alpha$, which are the {\sc sopde} condition for the $k$-vector field ${\bfX}$. Then, \eqref{primeq} holds identically, and \eqref{A-E-L-eqs3} and \eqref{A-E-L-eqs4} give
\begin{equation} 
    X_{\alpha\beta}^j\, \parderr{L}{v^j_\beta}{v^i_\alpha} + v_\alpha^j\, \parderr{L}{q^j}{v^i_\alpha} + X_\alpha^\beta\, \parderr{L}{s^\beta}{v^i_\alpha} - \parder{L}{q^i} = \parder{L}{s^\alpha} \parder{L}{v^i_\alpha}\,,\qquad  X_\alpha^\alpha = L \,. \label{burgos2} 
\end{equation} 
The following lemma is a direct consequence of equations \eqref{A-E-L-eqs2}, \eqref{A-E-L-eqs3} and \eqref{A-E-L-eqs4}.
\begin{lemma}\label{lema}
    Consider a Lagrangian function $L\in \Cinfty(\oplus^k\T Q\times\R^k)$.
    \begin{enumerate}[{\rm (1)}]
        \item If $L$ is a regular Lagrangian, then any $k$-vector 
        field ${\bfX} = (X_1,\dotsc,X_k)\in\X^k_L(\oplus^k\T Q\times\R^k)$ is a {\sc sopde}, which is locally given by formula        \eqref{localsode1} 
        and satisfies equations   \eqref{burgos2}. 
         Moreover, if $\bfX$ is integrable, its integral sections are canonical lifts $\phi_s^{(1)}$ of solutions of the $k$-contact Euler--Lagrange equations \eqref{eq:k-contact-Euler--Lagrange-equations}.
        
        \item If ${\bfX}$ is {\sc sopde} and ${\bfX}\in\X^k_L(\oplus^k\T Q\times\R^k)$ then it is locally given by formula        \eqref{localsode1} and satisfies equations   \eqref{burgos2}.
    \end{enumerate}
\end{lemma}

\begin{remark}
The results of Lemma \ref{lema} are the foundations of the $k$-contact Lagrangian formalism, and equations \eqref{eq:k-contact-Lagrangian-fields} can be seen as a geometric version of the $k$-contact Euler--Lagrange field equations \eqref{eq:eulerlagrange2}.
\end{remark}

\begin{remark}
    Notice that the particular case $k=1$ gives the contact Lagrangian formalism for mechanical systems with dissipation \cite{Gaset2020}.
\end{remark}

In the following example we look for solutions $(X_1,X_2) \in \mathfrak{X}^2_L (\oplus^2 \T \R \times \R^2)$ of the geometric $k$-contact Euler--Lagrange equations \eqref{eq:k-contact-Lagrangian-fields} for the regular Lagrangian describing the damped vibrating string, and we give an example of an integrable {\sc sopde} which is a solution.
\begin{example} 
    Let us consider again the vibrating string with damping introduced in Example \ref{exmpl:damped-vibrating-string}. Recall that the Lagrangian function $L:\oplus^2\T\R\times\R^2\to\R$ of this system given by \eqref{eq:Lagrangian-vibrating-string-dis} is regular. The Lagrangian energy associated to the Lagrangian function $L$ is
    $$ E_L = \Delta(L) - L = \frac{1}{2}\rho v_t^2 - \frac{1}{2}\tau v_x^2 + \gamma s^t\,, $$

    and the contact one-forms are
    \begin{equation*}
        \eta^t_L = \d s^t - \theta^t_L = \d s^t - \rho v_t\d q\,,\qquad\eta^x_L = \d s^x - \theta^x_L = \d s^x + \tau v_x\d q\,.
    \end{equation*}

    Consider now a two-vector field $\bfX = (X_1,X_2)\in\X^2(\oplus^2\T \R\times\R^2)$ with local expression
    \begin{align*}
        X_1 &= f_1\parder{}{q} + F_{1t}\parder{}{v_t} + F_{1x}\parder{}{v_x} + g_1^t\parder{}{s^t} + g_1^x\parder{}{s^x}\,,\\
        X_2 &= f_2\parder{}{q} + F_{2t}\parder{}{v_t} + F_{2x}\parder{}{v_x} + g_2^t\parder{}{s^t} + g_2^x\parder{}{s^x}\,.
    \end{align*}

    The first equation in \eqref{eq:k-contact-Lagrangian-fields} reads
    $$ \rho f_1\d v_t - \rho F_{1t}\d q - \tau f_2\d v_x + \tau F_{2x}\d q = \rho v_t\d v_t - \tau v_x\d v_x + \gamma\rho v_t\d q\,, $$
    which yields the conditions
    \begin{align}
        -\rho\, F_{1t} + \tau F_{2x} &= \gamma\rho v_t & & \mbox{(coefficients in $\d u$)}\,,\label{eq:vibrating-string-lagrangian-condition-1}\\
        f_1 &= v_t & & \mbox{(coefficients in $\d v_t$)}\,,\label{eq:vibrating-string-lagrangian-condition-2}\\
        f_2 &= v_x & & \mbox{(coefficients in $\d v_x$)}\,.\label{eq:vibrating-string-lagrangian-condition-3}
    \end{align}
    Notice that the equations \eqref{eq:vibrating-string-lagrangian-condition-2} and \eqref{eq:vibrating-string-lagrangian-condition-3} above are the {\sc sopde} conditions for the two-vector field $\bfX$.
    On the other hand, the second equation in \eqref{eq:k-contact-Lagrangian-fields} gives the condition
    $$ g_1^t + g_2^x = \frac{1}{2}\rho v_t^2 - \frac{1}{2}\tau v_x^2 - \gamma s^t = L\,. $$
    Hence, the two-vector field $\bfX$ solution has the local expression
    \begin{align*}
        X_1 &= v_t\parder{}{q} + \left( \frac{\tau}{\rho}F_{2x} - \gamma v_t \right)\parder{}{v_t} + F_{1x}\parder{}{v_x} + \left( L - g_2^x \right)\parder{}{s^t} + g_1^x\parder{}{s^x}\,,\\
        X_2 &= v_x\parder{}{q} + F_{2t}\parder{}{v_t} + F_{2x}\parder{}{v_x} + g_2^t\parder{}{s^t} + g_2^x\parder{}{s^x}\,,
    \end{align*}
    where the functions $F_{1x},F_{2t},F_{2x},g_1^x,g_2^t,g_2^x$ remain undetermined.
    
    In order to give an example of an integrable {\sc sopde} solution to \eqref{eq:k-contact-Lagrangian-fields}, we assume 
    \begin{itemize}
        \item the functions $F_{1x}=F_{2t}=g_1^x=g_2^t=0$,
        \item the functions $g_2^x$ only depend on the variables $v_x$,
        \item $\ds\frac{\partial F_2^x}{\partial s^x}=0$, $\ds\frac{\partial F_{2x}}{\partial v_t}\neq 0$, $\ds\frac{\partial F_{2x}}{\partial s^t} \neq 0$.
    \end{itemize}
    In this case, the integrability conditions, $[\Gamma_\alpha,\Gamma_\beta]=0$ for every $\alpha,\beta
    $, are
    $$\ds\frac{\partial F_{2x}}{\partial v_x}=0\,,\qquad \ds\frac{\partial g_2^x}{\partial v_x} = - \tau v_x\,, \qquad \ds\frac{\tau}{\rho} F_{2x} - \gamma v_t + L - g_2^x = 0 \,.$$
    Then, an integrable {\sc sopde} $(X_1,X_2)$ is
    $$\begin{array}{l}
    X_1 = v_t \, \ds\frac{\partial }{\partial q} + \left(- \frac{1}{2} \rho v_t^2 + \gamma s^t \right) \, \ds\frac{\partial }{\partial v_t} + \left(\frac{1}{2} \rho v_t^2 - \gamma s^t \right) \ds\frac{\partial}{\partial s^t} \, , \\ \noalign{\medskip}
    X_2 = v_x \, \ds\frac{\partial}{\partial q} + \left(\frac{\gamma \rho}{\tau} v_t - \frac{\rho^2}{2 \tau} v_t^2 + \frac{\gamma \rho}{\tau} s^t \right) \, \ds\frac{\partial}{\partial v_x} - \frac{1}{2} \tau v_x^2 \, \ds\frac{\partial}{\partial s^x}\, ,
    \end{array}$$
    and by Lemma \ref{lema}, any integral section of $(X_1,X_2)$ is a canonical lift of a solution $\phi_s$ of the $k$-contact Euler--Lagrange equations \eqref{eq:damped-vibrating-string-edp}.
\end{example}

\section{Relations between solutions and Lagrangian \texorpdfstring{$k$}--vector fields}\label{relations}

Now we characterize, locally and globally, the set of {\sc sopde}s in $\X^k_L(\oplus^k\T Q\times\R^k)$.

\begin{proposition}\label{prop:EL-sopde}
    Let $L\in\Cinfty(\oplus^k\T Q\times\R^k)$ be a Lagrangian function and $ {\bf \Gamma} \in \X^k(\oplus^k\T Q\times\R^k)$ a $k$-vector field. Then,
\begin{enumerate}[{\rm (1)}]
    \item 
    ${\bf \Gamma} \in\X^k_L(\oplus^k\T Q\times\R^k)$ if, and only if, it satisfies the conditions
    \begin{equation}\label{eq:lieeta} 
        \Lie_{\Gamma_\alpha} \eta^\alpha_L = \parder{L}{s^\alpha} \eta^\alpha_L\,,\qquad i_{\Gamma_\alpha}\eta^\alpha_L = -E_L \,.
    \end{equation} 
\item Moreover, if ${\bf \Gamma}$ is a {\sc sopde} on $\tkq \times \r^k$ with local expression \eqref{localsode1},
the above equations \eqref{eq:lieeta} can be written locally as
    \begin{equation}\label{eq:lieeta1} 
    \Gamma_\alpha\left(\parder{L}{v^i_\alpha}\right) - \parder{L}{q^i} = \parder{L}{s^\alpha}\parder{L}{v_\alpha^i} \,,\qquad \Gamma^\alpha_\alpha=L \,.
    \end{equation}
\end{enumerate}
\end{proposition}

\begin{proof}
\begin{enumerate}[{\rm (1)}]
    \item It is a direct consequence of equations \eqref{eq:k-contact-Lagrangian-fields} and Cartan calculus.

    \item Using the local expressions of the {\sc sopde} ${\bf \Gamma}$, the Lagrangian energy $E_L$, and the one-forms $\eta^\alpha_L$ (\eqref{localsode1}, \eqref{energy} and \eqref{eq:contact-forms} respectively),
    a direct computation on local coordinates shows that
    \begin{align*}
    0 & = i_{\Gamma_\alpha} \eta^\alpha_L + E_L = \left(\d s^\alpha - \parder{L}{v^i_\alpha}\d q^i\right)(\Gamma_\alpha) + v^i_\alpha \parder{L}{v^i_\alpha} - L 
     = \Gamma^\alpha_\alpha - L
    \end{align*}
    and
    \begin{align*}
        0 &= \Lie_{\Gamma_\alpha} \eta^\alpha_L - \parder{L}{s^\alpha} \eta_L^\alpha = -\d E_L + i_{\Gamma_\alpha}\d\eta^\alpha_L - \parder{L}{s^\alpha} \eta_L^\alpha \\
        &= -\d\left(v^i_\alpha \parder{L}{v^i_\alpha} -L\right)+ i_{\Gamma_\alpha} \left(\d q^i\wedge \d\left(\parder{L}{v^i_\alpha}\right)\right) - \parder{L}{s^\alpha} \left(\d s^\alpha - \parder{L}{v^i_\alpha}\d q^i\right)\\
        &= - \parder{L}{v^i_\alpha}\d v^i_\alpha + \parder{L}{q^i}\d q^i + \parder{L}{v^i_\alpha}\d v^i_\alpha + \parder{L}{s^\alpha}\d s^\alpha - \Gamma_\alpha\left(\parder{L}{v^i_\alpha}\right) \d q^i - \parder{L}{s^\alpha} \d s^\alpha + \parder{L}{s^\alpha} \parder{L}{v^i_\alpha}\d q^i\\
        &= \left( \parder{L}{q^i} - \Gamma_\alpha\left(\parder{L}{v^i_\alpha}\right) + \parder{L}{s^\alpha} \parder{L}{v^i_\alpha} \right)\d q^i\,,
    \end{align*}
    which proves the result.
 \end{enumerate}    
\end{proof}

We present now a new relation between solutions to the $k$-contact Euler--Lagrange equations \eqref{eq:k-contact-Euler--Lagrange-equations} and $k$-contact Lagrangian $k$-vector fields on $\oplus^k\T Q\times\R^k$. This relation plays a fundamental role in this paper.
\begin{proposition}\label{known2}
    Let $L\in \Cinfty(\oplus^k\T Q\times\R^k)$ be a Lagrangian function.
\begin{enumerate}[{\rm (1)}]
\item A map 
$\phi_s: U \subset \R^k\to Q\times\R^k \, ,\, 
  \phi_s(t)=(\phi(t),s^\alpha(t))$
  is a solution to the
  $k$-contact Euler--Lagrange equations \eqref{eq:eulerlagrange2} if, and only if,
\begin{equation}
g^{\alpha\beta}_{ij}\circ \phi_s^{(1)} \left( \Gamma^j_{\alpha\beta}\circ
  \phi_s^{(1)} - \frac{\partial^2 \phi^j}{\partial t^{\alpha}\partial t^{\beta}}\right)+ \frac{\partial^2 L}{\partial s^\beta\partial v^i_\alpha}\circ \phi_s^{(1)}\left( \Gamma_\alpha^\beta\circ \phi_s^{(1)}  - \derpar{s^\beta}{t^\alpha}\right)= 0\,,
     \label{eqn:xiel0}
     \end {equation}
     \begin{equation}
\Gamma^\alpha_\alpha\circ \phi_s^{(1)}= \derpar{s^\alpha}{t^\alpha} = L\circ \phi_s^{(1)}\,, 
 \label{eqn:xiel}
\end{equation}
for any {\sc sopde} ${\bf \Gamma} \in\X^k_L(\oplus^k\T Q\times\R^k)$, with local expression \eqref{localsode1}.
\item If a $k$-vector field ${\bf \Gamma} \in\X^k_L(\oplus^k\T Q\times\R^k)$ is integrable,
and $\phi_s^{(1)}: U \subset \R^k \to       \oplus^k\T Q\times\R^k$ is an integral section,
then $\phi_s$ is a solution to the $k$-contact Euler--Lagrange equations  \eqref{eq:eulerlagrange2}.
\end{enumerate}
\end{proposition}
\begin{proof}
\begin{enumerate}[{\rm (1)}]
    \item Consider a map $\phi_s: U \subset \R^k\to Q\times\R^k \,,\,\,
  \phi_s(t)=(\phi(t),s^\alpha(t))\,$. If $\phi_s$ is a solution to the
   $k$-contact Euler--Lagrange equations
  \eqref{eq:eulerlagrange2},
  then we have
  \begin{equation}
 g^{\alpha\beta}_{ij}\circ \phi_s^{(1)} \frac{\partial^2 \phi^j}{\partial t^{\alpha}\partial
  t^{\beta}} + \frac{\partial^2L}{\partial q^j \partial v^i_{\alpha}} \circ \phi_s^{(1)} \frac{\partial \phi^j}{\partial t^\alpha}
  +\frac{\partial^2 L}{\partial s^\beta\partial v^i_\alpha}\circ \phi_s^{(1)}\derpar{s^\beta}{t^\alpha}
- \frac{\partial L}{\partial q^i}\circ \phi_s^{(1)}=\left(\displaystyle\frac{\partial L}{\partial s^\alpha}\displaystyle\frac{\partial L}{\partial v^i_\alpha}\right)\circ \phi_s^{(1)}\,, \label{gxiphi0}
 \end{equation}
 \begin{equation}
\derpar{s^\alpha}{t^\alpha} = L\circ \phi_s^{(1)}\,. \label{gxiphi1}  
 \end{equation}

Using Proposition \ref{prop:EL-sopde}, a {\sc sopde} ${\bf \Gamma}$ is in $\X^k_L(\oplus^k\T Q\times\R^k)$  if, and only if, it satisfies equations equations
\begin{equation}\label{burgos10}
\Gamma_\alpha^\alpha=L\,,
\end{equation}
\begin{equation}\label{burgos20}
\frac{\partial^2 L}{\partial v^j_\beta\partial v^i_\alpha}\Gamma_{\alpha\beta}^j
+\frac{\partial^2 L}{\partial q^j \partial v^i_\alpha}v_\alpha^j
+\frac{\partial^2 L}{\partial s^\beta\partial v^i_\alpha}\Gamma_\alpha^\beta
\displaystyle
-\frac{\partial L}{\partial q^i} =\frac{\partial L}{\partial s^\alpha}
\frac{\partial L}{\partial v^i_\alpha}\,.
\end{equation}

Now, restricting equations \eqref{burgos10} and \eqref{burgos20} to the image of $\phi_s^{(1)}=(\phi^{(1)},s^\alpha)$ one gets
\begin{equation}
     (g^{\alpha\beta}_{ij}\circ \phi_s^{(1)}) (\Gamma^j_{\alpha\beta}\circ
  \phi_s^{(1)}) + \frac{\partial^2 L}{\partial q^j \partial v^i_{\alpha}}
  \circ \phi_s^{(1)}+ \frac{\partial \phi^j}{\partial t^{\alpha}}\frac{\partial^2 L}{\partial s^\beta\partial v^i_\alpha}\circ \phi_s^{(1)}(\Gamma_\alpha^\beta \circ \phi_s^{(1)}) - \frac{\partial L}{\partial q^i}
  \circ \phi_s^{(1)}
  =\left(\frac{\partial L}{\partial s^\alpha}
\frac{\partial L}{\partial v^i_\alpha}\right)\circ \phi_s^{(1)}\,,  \label{gxiphi00}
\end{equation}
\begin{equation}
\Gamma_\alpha^\alpha\circ \phi_s^{(1)}=L \circ \phi_s^{(1)}.\label{gxiphi}\end{equation}

Using equations \eqref{gxiphi00} and \eqref{gxiphi}, we have that $\phi_s$ satisfies \eqref{eqn:xiel0}
and \eqref{eqn:xiel} if, and only if, it satisfies \eqref{gxiphi0} and \eqref{gxiphi1}, that are equivalent to $k$-contact Euler--Lagrange equations \eqref{eq:eulerlagrange2}.

\item Since ${\bf \Gamma} \in \X^k_L(\tkq
 \times \r^k)$ it follows that ${\bf \Gamma}$ satisfies equations \eqref{primeq}, \eqref{A-E-L-eqs2}, \eqref{A-E-L-eqs3} and \eqref{A-E-L-eqs4}.
 If $\phi^{(1)}_s$ is an integral map of ${\bf \Gamma}$, we know that
 $\Gamma_\alpha^i(\phi^{(1)}_s) = v^i_\alpha(\phi^{(1)}_s) ,$
 and if we restrict equations \eqref{A-E-L-eqs3} to $\phi^{(1)}_s$, we obtain that $\phi_s$ satisfies the $k$-contact Euler--Lagrange equations \eqref{eq:eulerlagrange2}.
\end{enumerate}
\end{proof}

\begin{remark}
Equations \eqref{eqn:xiel0} and \eqref{eqn:xiel} do not require any
relationship between the $k$-vector field ${\bf \Gamma} \in\X^k_L(\oplus^k\T Q\times\R^k)$ and the solution $\phi_s$ to the $k$-contact Euler--Lagrange equations
\eqref{eq:eulerlagrange2}. In other words, one might
  have a solution $\phi_s$ to the $k$-contact Euler--Lagrange equations
  \eqref{eq:eulerlagrange2} which may not be a  solution for any ${\bf \Gamma} \in\X^k_L(\oplus^k\T Q\times\R^k)$.
\end{remark}

\begin{remark}
Propositions \ref{prop:EL-sopde} and \ref{known2} play an important role in subsequent sections of this paper. 
When the Lagrangian $L$ does not depend on the variables $s^\alpha$, for $\alpha=1,\ldots,k$, the function $L$ may be defined on the bundle $\tkq$, and the previous results can be formulated in terms of the $k$-symplectic Lagrangian formalism (see Propositions 2.11 and 2.12 in \cite{BBS}).
\end{remark}

\section{Dissipation laws}\label{dlcs}

In this section, we discuss dissipation laws for Lagrangian functions defined on $\tkq \times \r^k$ and we give certain relations between them and integrable {\sc sopde}s on $\X^k_L( \oplus^k\T Q\times\R^k)$. From now on, unless otherwise stated, $L$ will denote a (regular or singular) Lagrangian function on the phase bundle $\oplus^k\T Q\times\R^k$ and $\eta_L^\alpha$ will denote the associated contact Lagrangian one-forms.

\begin{definition}\label{dynsym}
    A map $F=(F^1,\ldots, F^k)\colon \oplus^k\T Q\times\R^k \to \R^k$ is called a \textbf{dissipation law} if the divergence of
    $$F\circ\phi_s^{(1)}=(F^1\circ\phi_s^{(1)},\ldots,F^k\circ\phi_s^{(1)})\colon U\subset \R^k \longrightarrow \R^k$$ 
    satisfies
    \begin{equation}\label{disipkmaps}
        \Div(F\circ\phi_s^{(1)}) = \derpar{(F^\alpha\circ\phi_s^{(1)})}{t^\alpha}= \left( \ds\derpar{L}{s^\alpha}F^\alpha \right) \circ\phi_s^{(1)}\,,
    \end{equation} 
    for every solution $\phi_s\colon U\subset  \R^k\to Q\times \R^k$ of the $k$-contact Euler--Lagrange equations \eqref{eq:k-contact-Euler--Lagrange-equations}.     
\end{definition}

Then, we have
 \begin{align}
 \frac{\partial (F^\alpha \circ \phi^{(1)}_s)}{\partial
t^\alpha}\Big\vert_{t} &= \frac{\partial F^{\alpha}}{\partial
q^i}\Big\vert_{\phi_s^{(1)}(t)}\frac{\partial \phi^i}{\partial
t^{\alpha}}\Big\vert_{ t} +\frac{\partial F^{\alpha}} {\partial
v^i_{\beta}}\Big\vert_{ \phi_s^{(1)}(t)} \frac{\partial^2
\phi^i}{\partial t^{\alpha}
\partial t^{\beta}}\Big\vert_{t} +\frac{\partial F^{\alpha}} {\partial
s^{\beta}}\Big\vert_{ \phi_s^{(1)}(t)} \derpar{s^\beta}{t^\alpha}\Big\vert_{t}\nonumber \\
&= \left( \ds\derpar{L}{s^\alpha}F^\alpha  \right)\circ\phi_s^{(1)}(t)\,,\label{conlaw}
\end{align}
where $\phi_s^{(1)}(t)=\left(\phi^i(t), \dparder{\phi^i}{t^\alpha }\Big\vert_t,s^\alpha(t)\right)$.

In the following example, we present two dissipation laws for the equation of the damped vibrating string given by the Lagrangian function \eqref{eq:Lagrangian-vibrating-string-dis}.
\begin{example}
The two maps $F_1 = (F_1^t, F_1^x), F_2 = (F_2^t, F_2^x) \colon \oplus^2 \T \R\times \R^2\to\R^2$, where
   \begin{itemize} 
       \item [a)] $F_1^t = \rho v_t ,\quad F_1^x = -\tau v_x $,
       \item [b)] $F_2^t = \ds\frac{1}{\rho} s^t - \ds\frac{1}{2} q v_t , \quad F_2^x = \ds\frac{1}{\rho} s^x + \ds\frac{\tau}{2\rho} q v_x $,
   \end{itemize}
give dissipation laws for the $k$-contact Euler--Lagrange equations \eqref{eq:damped-vibrating-string-edp}.
Thus, if $\phi_s$ is a solution to \eqref{eq:damped-vibrating-string-edp}, we deduce that
    $$\frac{\partial (F_1^t \circ \phi^{(1)}_s)}{\partial
t} + \frac{\partial (F_1^x \circ \phi^{(1)}_s)}{\partial
x} = \rho \frac{\partial^2 \phi}{\partial t^2} - \tau \frac{\partial^2 \phi}{\partial x^2} = - \gamma \rho \frac{\partial \phi}{\partial t} = \left(\parder{L}{s^t}F_1^t\right) \circ \phi_s^{(1)} ,$$
where $\phi^{(1)}_s=\left(\phi,\dparder{\phi}{t},\dparder{\phi}{x},s^t,s^x\right)$.

One also verifies that $(F_2^t,F_2^x)$ satisfies \eqref{conlaw} by a straightforward computation.
\end{example}

\begin{remark}
Note that if the Lagrangian $L$ does not depend on the variables $s^\alpha$, for $\alpha= 1,\ldots,k$, we can consider $L$ and the functions $F^\alpha$ defined on $\oplus^k\T Q $, and hence the above Definition \ref{dynsym} becomes the definition of conservation law for $k$-symplectic Lagrangian systems (see \cite{rsv07}).
\end{remark}

The following result gives us a first relationship between dissipation laws and integrable {\sc sopde}s on $\X^k_L( \oplus^k\T Q\times\R^k)$.

\begin{lemma}\label{imp} Let $F=(F^1,\ldots,F^k)\colon\oplus^k\T Q\times\R^k\rightarrow\R^k$ be a dissipation law. Then, every integrable {\sc sopde} ${\bf \Gamma} = (\Gamma_1,\dots,\Gamma_k)\in \X^k_L(\oplus^k\T Q\times\R^k)$ satisfies
\begin{equation}\label{xifa}
    \Gamma_{\alpha}(F^\alpha) = \parder{L}{s^\alpha}F^\alpha\,.
\end{equation}
 \end{lemma}
\begin{proof}
Since ${\bf \Gamma}$ is an integrable {\sc sopde}, for every point $x\in\oplus^k\T Q\times\R^k$ there exists an integral
section $\phi_s^{(1)}:U\subset\R^k\rightarrow\oplus^k\T Q\times\R^k$ such
that
\begin{enumerate}[{\rm (1)}]
\item $\phi_s$ is a solution to the $k$-contact Euler--Lagrange equations, because ${\bf \Gamma} \in     \X^k_L( \oplus^k\T Q\times\R^k)$,
\item $\phi_s$ satisfies
$$ \phi_s^{(1)}(0)=x\,, \qquad
(\phi_s^{(1)})_*(t)\left(\parder{}{t^\alpha}\Big\vert_t\right) = \Gamma_{\alpha}(\phi_s^{(1)} (t))\,,
$$
for every $t\in U$ and $\alpha=1,\ldots, k$.
\end{enumerate}
Condition (2) above means that
\begin{equation}\label{caz}
 v^i_\alpha(\phi_s^{(1)}(t))=\frac{\partial \phi^i} {\partial
t^\alpha}\Big\vert_{t} \,, \qquad  \Gamma^i_{\alpha
\beta}(\phi_s^{(1)}(t)) =\frac{\partial^2 \phi^i}{\partial t^{\alpha}
\partial t^{\beta}}\Big\vert_{t}  \,, \qquad   \Gamma_\alpha^\beta(\psi(t))=\frac{\partial s^\alpha} {\partial
t^\beta}\Big\vert_t\,.
\end{equation} 

Since $F=(F^1,\ldots,F^k)$ is a dissipation
law, using equation \eqref{conlaw} at $t=0$, and equation \eqref{caz}, we have
\begin{align}
 \left( \ds\derpar{L}{s^\alpha}F^\alpha \right) \circ\phi_s^{(1)}(0) &= \ds \frac{\partial (F^\alpha \circ \phi_s^{(1)})}{\partial
t^\alpha}\Big\vert_{0} \\ 
&= \ds\frac{\partial
F^{\alpha}}{\partial q^i}\Big\vert_{ \phi_s^{(1)}(0)} \frac{\partial
\phi^i}{\partial t^{\alpha}}\Big\vert_{ 0} +\frac{\partial
F^{\alpha}} {\partial v^i_{\beta}}\Big\vert_{ \phi_s^{(1)}(0)}
\frac{\partial^2 \phi^i}{\partial t^{\alpha}
\partial t^{\beta}}\Big\vert_{0} +\frac{\partial F^{\alpha}} {\partial
s^{\beta}}\Big\vert_{ \phi_s^{(1)}(0)} \derpar{s^\beta}{t^\alpha}\Big\vert_{0}
\\
&=
\ds\frac{\partial F^{\alpha}}{\partial q^i}\Big\vert_{x}
v^i_\alpha(x) +\ds\frac{\partial F^{\alpha}} {\partial
v^i_{\beta}}\Big\vert_{x} \Gamma^i_{\alpha\beta}(x) 
+\ds\frac{\partial F^{\alpha}} {\partial
s^\beta}\Big\vert_{x} \Gamma_\alpha^\beta(x)
= \Gamma_{\alpha}(x)(F^\alpha)\,.
\end{align}
\end{proof}

The converse of Lemma \ref{imp} is not true, and the reason is that, as we can see from equations \eqref{eqn:xiel0} and \eqref{eqn:xiel}, we might have solutions $\phi_s$ of the $k$-contact Euler--Lagrange equations \eqref{eq:eulerlagrange2}  which are not solutions to some ${\bf \Gamma} \in \X^k_L(\oplus^k\T Q\times\R^k)$. However, we show in the following lemma that, under some assumption on the functions $F^{\alpha}$, this converse is true.

\begin{lemma} \label{lem:12}
Let $L\in \Cinfty(\oplus^k\T Q\times\R^k)$ be a Lagrangian and
 assume that there exists  a vector field $X\in
\X(\oplus^k\T Q\times\R^k)$ such that
\begin{equation}\label{ixoa}  i_X \d\eta_L^\alpha = \d F^\alpha\,,\quad \text{for every} \,\,\, \alpha=1,\ldots,k\,,
\end{equation}
for some functions $F^{\alpha}:\oplus^k\T Q\times\R^k\to\R$. Then, $F = (F^{\alpha})$ is a dissipation law for the $k$-contact Euler--Lagrange equations
\eqref{eq:k-contact-Euler--Lagrange-equations} if, and only if,
$$\Gamma_{\alpha}(F^{\alpha})=\frac{\partial L}{\partial s^\alpha} F^\alpha$$ for every integrable {\sc sopde}   ${\bf \Gamma} \in \X^k_L(      \oplus^k\T Q\times\R^k)$.
\end{lemma}
\begin{proof}
 The direct implication is given by Lemma \ref{imp}.
 For the converse implication, let $$ X=X^i\dparder{}{q^i} + X^i_{\alpha}
 \dparder{}{v^i_{\alpha}}
 +X^{\alpha} \dparder{}{s^{\alpha}}$$
 be a vector field on $\oplus^k\T Q\times\R^k$  satisfying \eqref{ixoa}.  In view of equations \eqref{eq:contact-forms}, we can
write both sides of the equations \eqref{ixoa} as
\begin{align}
i_X \d\eta_L^\alpha&=\left[\left(\frac{\partial^2L}{\partial q^i\partial v^j_{\alpha}} -
    \frac{\partial^2L}{\partial q^j \partial v^i_{\alpha}}\right) X^j
  - g^{\alpha\beta}_{ij} X^j_{\beta} - \frac{\partial^2L}{\partial s^\beta \partial v^j_\alpha} X^\beta \right]\d q^i +
g^{\alpha\beta}_{ij}X^i \d v^j_{\beta} + \frac{\partial^2L}{\partial s^\beta \partial v^j_\alpha} X^j \d s^\beta\,,\\
\d F^\alpha &= \frac{\partial
	F^{\alpha}}{\partial q^i} \d q^i +  \frac{\partial
	F^{\alpha}}{\partial v^j_{\beta}} \d v^j_{\beta} + \frac{\partial F^\alpha}{\partial s^\beta} \d s^\beta\,,
\end{align}
and necessarily we have
\begin{equation}
 \frac{\partial F^{\alpha}}{\partial v^j_{\beta}} =
 g^{\alpha\beta}_{ij} X^i \,, \qquad \frac{\partial F^{\alpha}}{\partial s^{\beta}} =
 \frac{\partial^2L}{\partial s^\beta \partial v^i_\alpha} X^i \,. \label{pfg} 
\end{equation}
Consider now $\phi_s$ to be any solution to the $k$-contact Euler--Lagrange equations
\eqref{eq:eulerlagrange2} (which may not be a solution of any ${\bf \Gamma}$). It
follows that $\phi_s$ satisfies equations \eqref{eqn:xiel0} and \eqref{eqn:xiel}, since ${\bf \Gamma}$ is
assumed to be an integrable {\sc sopde}.

Contracting equations \eqref{eqn:xiel} with  $X^i\circ \phi_s^{(1)}$, we obtain
\begin{equation}
(X^i\circ \phi_s^{(1)}) (g^{\alpha\beta}_{ij}\circ \phi_s^{(1)}) \left(\Gamma^j_{\alpha\beta}\circ \phi_s^{(1)} - \frac{\partial^2 \phi^j}{\partial t^{\alpha} \partial t^{\beta}}\right) + (X^i\circ \phi_s^{(1)}) \left(\frac{\partial^2 L}{\partial s^\beta\partial v^i_\alpha}\circ \phi_s^{(1)}\right) \left( \Gamma_\alpha^\beta\circ \phi_s^{(1)}  -    
  \derpar{s^\beta}{t^\alpha}    \right) = 0 .
  \label{xgphi}
\end{equation}
If we replace formulas \eqref{pfg} in equation \eqref{xgphi}, we have
\begin{align}
0 & = \ds\frac{\partial F^{\alpha}}{\partial
  v^j_{\beta}}\circ \phi_s^{(1)}\left(\Gamma^j_{\alpha\beta} \circ
\phi_s^{(1)} - \ds\frac{\partial^2 \phi^j}{\partial
    t^{\alpha}\partial t^{\beta}} \right) + \ds\frac{\partial F^\alpha}{\partial s^\beta}\circ \phi_s^{(1)}\left( \Gamma_\alpha^\beta\circ \phi_s^{(1)}  -    
\ds\parder{s^\beta}{t^\alpha}    \right) \\
&= - \ds\frac{\partial (F^{\alpha} \circ \phi_s^{(1)})}{\partial
  t^{\alpha}} + \Gamma_{\alpha}(F^{\alpha}) \circ \phi_s^{(1)} \, ,
\end{align}
therefore we conclude
$$\frac{\partial (F^{\alpha} \circ \phi_s^{(1)})}{\partial
	t^{\alpha}} = \Gamma_{\alpha}(F^{\alpha}) \circ \phi_s^{(1)} = \left(\frac{\partial L}{\partial s^\alpha} F^\alpha\right) \circ \phi_s^{(1)}\,,$$
and the result follows.
\end{proof}

\begin{remark}
Notice that the above Lemma \ref{lem:12} is analogous  to Lemma 3.4 in \cite{BBS} for the $k$-symplectic Lagrangian formalism.
\end{remark}

\section{Symmetries}\label{section:symmetries}

In this section we study several different notions of symmetry for $k$-contact Lagrangian field theories depending on the structure they preserve: natural symmetries (preserving the Lagrangian), dynamical symmetries (preserving the solutions) and $k$-contact symmetries (preserving the underlying geometric structures). We investigate the relations between these symmetries and prove several Noether-like theorems relating the different types of symmetries to dissipation laws.

\subsection{Natural symmetries}

We will now try to understand the symmetries of a $k$-contact Lagrangian system which are lifts of vector fields on the configuration space. One can see in \cite{LML} some interesting results related with this kind of symmetries in the case $k=1$.

For any vector field $Z\in\X(Q)$, we denote by $Z^C , Z^{V_\alpha}\in \X(\oplus^k\T Q\times\R^k)$ its complete and vertical $\alpha$-lifts, defined in Section \ref{VClifts} and extended to $\oplus^k\T Q\times\R^k$ in a natural way.
Their local expressions are the same as in
\eqref{lifts-vectors1}.

\begin{definition}
A vector field $Z\in \mathfrak{X}(Q)$ is said to be an \textbf{infinitesimal natural symmetry} of a Lagrangian function $L\in\Cinfty(\oplus^k\T Q\times\R^k)$ if  $Z^C(L)=0$.
\end{definition}
Then, we have the following result relating natural symmetries and dissipation laws.
\begin{theorem}
 Let $Z \in \mathfrak{X}(Q)$ be an infinitesimal natural symmetry of a Lagrangian function $L$. Then, the functions $F^\alpha=Z^{V_\alpha}(L)$ give a dissipation law.
\end{theorem}

\begin{proof}
Consider the local expressions of $Z^{V_\alpha}$, $Z^C$ and ${\bf \Gamma}$ in \eqref{lifts-vectors1} and \eqref{localsode1}.
From \eqref{eq:lieeta1}, and  taking into account that the functions $Z^i$ only depend on the variables $q^i$, we get
\begin{equation}\label{infnatsym}
\Gamma_\alpha(Z^{V_\alpha}(L)) 
= \left(Z^i\frac{\partial}{\partial q^i} + v^j_\alpha \ds\frac{\partial Z^i}{\partial q^j} \frac{\partial}{\partial v^i_\alpha}\right)(L) + \ds\frac{\partial L}{\partial s^\alpha} \left(Z^i \frac{\partial}{\partial v^i_\alpha}\right)(L) 
 = Z^C(L) + \ds\frac{\partial L}{\partial s^\alpha} Z^{V_\alpha} (L) \,,
\end{equation}
for all integrable {\sc sopde} ${\bf \Gamma} \in \X^k_L(\oplus^k\T Q\times\R^k)$.

 Now, since $Z^C(L)=0$ we show that $\Gamma_{\alpha}(Z^{V_\alpha}(L))=\dparder{L}{s^\alpha} Z^{V_\alpha}(L)$. Finally, in view of Lemma \ref{lem:12}, we just need to check that \begin{equation}\label{eqYC}
    i_{Z^C} \d\eta^\alpha_L = \d\left(Z^{V_\alpha} (L)\right)
\end{equation}
to conclude that the functions $Z^{V_\alpha} (L)$ give a dissipation law.

If one applies $\dfrac{\partial}{\partial v^k_\beta}$ to the identity $Z^C(L)= Z^i \dfrac{\partial L}{\partial q^i} + v^j_\alpha \dfrac{\partial Z^i}{\partial q^j} \dfrac{\partial L}{\partial v^i_\alpha} = 0$,
 we obtain the relation \eqref{eqYC} in local coordinates.
\end{proof}
\begin{remark}
When the Lagrangian $L$ does not depend on variables $s^\alpha$, for $\alpha=1,\dotsc,k$, the function $L$ may be defined on the bundle $\tkq$, and the above result is similar to Proposition 3.15 in \cite{rsv07}, if we consider the function $g$ to be identically zero. In this case, the functions $F^\alpha$ will give a conservation law, according to the nature of the system.
\end{remark}

The next symmetries we are going to study are those transformations preserving some structure, and we will distinguish  between those that preserve the solutions of the system (dynamical symmetries) and those that preserve the geometric structures ($k$-contact symmetries). These notions were introduced in \cite{Gaset2020} for the case $k=1$, and in \cite{Gaset2021} for the general case $k\geq1$.

 \subsection{Dynamical symmetries}
 
We will begin by introducing the transformations preserving the solutions of the system.
\begin{definition} \label{dinamica-difeo}
A \textbf{Lagrangian dynamical symmetry} is a diffeomorphism 
$ \Phi: \oplus^k\T Q\times\R^k\to \oplus^k\T Q\times\R^k$  such that, for every solution $\phi_s^{(1)}$  to the $k$-contact Euler--Lagrange
equations \eqref{eq:k-contact-Euler--Lagrange-equations}, $\Phi\circ  \phi_s^{(1)}$ is also a solution.

An \textbf{infinitesimal Lagrangian dynamical symmetry} is a vector field $X\in \X(\oplus^k\T Q\times\R^k)$ whose local flow is made of dynamical symmetries.
 \end{definition}

The following result can be found in \cite{Gaset2021}.
\begin{lemma}\label{lemadynsim}
Let $X \in \mathfrak{X}(\oplus^k \T Q \times \R^k)$ be an infinitesimal dynamical symmetry. Then, for every ${\bf \Gamma} \in \mathfrak{X}^k_L(\oplus^k \T Q \times \R^k)$, we have
$i_{[X,\Gamma_\alpha]} \eta^\alpha_L = 0 .$
\end{lemma}

For the case $k=1$ it is known that dynamical symmetries induce dissipated quantities, and these results are known as
dissipation theorems (see \cite{Gaset2020}).

In the general case $k>1$, in Theorem 3 in \cite{Gaset2021}, it is proved that if $X$ is an infinitesimal dynamical symmetry, the functions $F^{\alpha}=- i_X \eta^\alpha_L$ satisfy \eqref{xifa}, and give a dissipation law for  integral sections of ${\bf \Gamma} \in \mathfrak{X}^k_L(\oplus^k \T Q \times \R^k)$.
     
The next result states, again with an extra condition, that 
these functions $F^{\alpha}$ are a dissipation law for every solution of the $k$-contact Euler--Lagrange  equations \eqref{eq:k-contact-Euler--Lagrange-equations}.

\begin{theorem}[Dissipation theorem]\label{thm:noether}
Let $L\in \Cinfty(\oplus^k\T Q\times\R^k)$ be a Lagrangian function and consider an infinitesimal dynamical symmetry $X\in \mathfrak{X}(\oplus^k\T Q\times\R^k)$ satisfying the conditions
\begin{equation}\label{hypdisipth}
    i_X \d\eta^\alpha_L =  \d(- i_X \eta^\alpha_L)\,,
    \end{equation}
    or, equivalently, $\Lie_X \eta^\alpha_L = 0$. Then, the functions
	\begin{equation} 
	F^{\alpha}=- i_X \eta^\alpha_L \label{conslaw0}
	\end{equation}
	provide a dissipation law for the $k$-contact Euler--Lagrange equations.
\end{theorem}
\begin{proof}
Let ${\bf \Gamma} \in \mathfrak{X}^k_L(\oplus^k \T Q \times \R^k)$ be an integrable {\sc sopde}. From Lemma \ref{lemadynsim}, we obtain
$$\Gamma_\alpha(- i_X \eta^\alpha_L) = - \Lie_{\Gamma_\alpha} i_X \eta^\alpha_L = - i_X \Lie_{\Gamma_\alpha} \eta^\alpha_L - i_{[\Gamma_{\alpha},X]} \eta^\alpha_L = \ds\frac{\partial L}{\partial s^\alpha} (- i_X \eta^\alpha_L) .$$
Now, from \eqref{hypdisipth} and Lemma \ref{lem:12}, we deduce that the functions $F^\alpha = - i_X \eta^\alpha_L$ give a dissipation law.
\end{proof}

\subsection{\texorpdfstring{$k$}--contact symmetries}
  
Among the most relevant symmetries are those that leave the geometric structures invariant. Let us recall the definition of $k$-contact symmetries for a Lagrangian function $L\in\Cinfty(\oplus^k\T Q\times\R^k)$ (see \cite{rsv07}).

\begin{definition}
    A diffeomorphism $\Phi:\oplus^k\T Q\times\R^k\to \oplus^k\T Q\times\R^k$ is called a \textbf{Lagrangian $k$-contact symmetry} if
    $$\Phi^*\eta^{\alpha}_L=\eta^{\alpha}_L\,,\qquad\Phi^*E_L=E_L\,. $$
    A vector field $ X\in
    \mathfrak{X}(\oplus^k\T Q\times\R^k)$ is called an \textbf{infinitesimal Lagrangian $k$-contact symmetry}  if
\begin{equation}\label{last}
    \Lie_X\eta^{\alpha}_L=0\,, \qquad \Lie_X E_L=0\,,
\end{equation}    
 that is, its local flow is made of Lagrangian $k$-contact symmetries.
\end{definition}

One can obtain the relation between these symmetries that preserve geometric structures and those that preserve the solutions of the system as follows.

\begin{proposition}\label{contact-dinamica}
  (Infinitesimal) Lagrangian $k$-contact symmetries are (infinitesimal)
Lagrangian dynamical symmetries.
\end{proposition}

\begin{proof}
Let $\psi=\phi_s^{(1)}$ be a solution of \eqref{eq:k-contact-Euler-Lagrange-section}, 
and $\Phi$ a Lagrangian $k$-contact symmetry
\begin{equation}\label{k-contsym}
   \Phi^*\eta^{\alpha}_L=\eta^{\alpha}_L\,, \qquad \Phi^*E_L=E_L\,.
\end{equation}
We need to prove that  
\begin{equation}\label{saturday}
i_{(\Phi\circ\psi)'_\alpha}\d\eta^\alpha_L = \left( \d E_L  +
\dparder{L}{s^\alpha} \eta^\alpha_L \right)\circ(\Phi\circ \psi)\,,\qquad
i_{(\Phi\circ\psi)'_\alpha}\eta^\alpha_L = -E_L\circ(\Phi\circ \psi)\,,
\end{equation}
where $ (\Phi\circ\psi)'_\alpha=(\Phi\circ \psi)_*(t)(\parder{}{t^\alpha}) \, ,$
which is a consequence of \eqref{eq:k-contact-Euler-Lagrange-section} and \eqref{k-contsym}.

If $X$ is an infinitesimal $k$-contact symmetry, its flow $\phi_t$ is made of local $k$-contact symmetries, namely $\phi_t^*\eta^{\alpha}_L=\eta^{\alpha}_L$ and $\phi_t^*E_L=E_L$.
Thus $\phi_t$ satisfies \eqref{saturday}, and then it transforms solutions to the $k$-contact Euler--Lagrange equations on solutions to the $k$-contact Euler--Lagrange equations.
\end{proof}

From Dissipation Theorem \ref{thm:noether} and 
Proposition \ref{contact-dinamica} we deduce that
\begin{corollary}\label{cafe}
If $X$ is an infinitesimal Lagrangian $k$-contact symmetry, the functions $F^\alpha= -i_X \eta^\alpha_L$ give a dissipation law for the $k$-contact Euler--Lagrange equations \eqref{eq:k-contact-Euler--Lagrange-equations}.
\end{corollary}

\begin{example}\label{examplekcontsym}
Consider again Example \ref{exmpl:damped-vibrating-string} (the damped vibrating string). The vector field $X= \dfrac{\partial}{\partial q}$ is an infinitesimal Lagrangian $k$-contact symmetry, and therefore, the map given by
$$ F = \left(-i_{\parder{}{q}} \eta^t_L, -i_{\parder{}{q}} \eta^x_L\right) = \left( \rho v_t, - \tau v_x\right) $$
is a dissipation law.
\end{example}

\subsection{Dissipation laws given by vector fields which are not symmetries}

Finally, let us consider any vector field $Z\in \X(\oplus^k\T Q\times\R^k)$ satisfying the conditions
\begin{equation}\label{Cartnsym}
 \Lie_Z\eta^{\alpha}_L=  \d g^\alpha\,,\quad  \quad \Lie_ZE_L = - g^\alpha \ds\frac{\partial L}{\partial s^\alpha}\,,  
\end{equation}
for some functions $g^\alpha \in \Cinfty(\oplus^k\T Q\times\R^k)$. In the following result we show how to associate dissipation laws to these particular vector fields.

\begin{theorem}\label{disipthc}
Let $L:\oplus^k\T Q\times\R^k\to\R$ be a Lagrangian function and consider a vector field $Z \in \mathfrak{X}(\oplus^k \T Q \times \R^k)$ satisfying conditions \eqref{Cartnsym}. Then, the functions
$$ F^\alpha= g^\alpha - i_Z \eta^\alpha_L $$
give a dissipation law for the $k$-contact Euler--Lagrange equations.
\end{theorem}

\begin{proof}
Consider an integrable {\sc sopde} ${\bf \Gamma} \in \mathfrak{X}^k_L(\oplus \T Q \times \R^k)$. Then, by Proposition \ref{prop:EL-sopde}, the {\sc sopde} ${\bf \Gamma}$ satisfies equations \eqref{eq:lieeta}, so we compute
\begin{align}\label{eqproof1}
\begin{split}
 \Gamma_\alpha(F^\alpha) &= \Gamma_\alpha(g^\alpha) - \Lie_{\Gamma_\alpha} i_Z \eta^\alpha_L = \Gamma_\alpha(g^\alpha) - i_Z \Lie_{\Gamma_\alpha} \eta^\alpha_L - i_{[\Gamma_\alpha,Z]} \eta^\alpha_L \\
 &= \Gamma_\alpha(g^\alpha) - \ds\frac{\partial L}{\partial s^\alpha} i_Z \eta^\alpha_L - i_{[\Gamma_\alpha,Z]} \eta^\alpha_L \,.   
 \end{split}
\end{align}
Contracting $\Lie_Z \eta^\alpha_L$ with $\Gamma_\alpha$ and summing for $\alpha$, we have
\begin{equation}\label{eqproof2}
i_{\Gamma_\alpha} \Lie_Z \eta^\alpha_L = i_{\Gamma_\alpha} (\d g^\alpha) =  \Gamma_\alpha(g^\alpha)\,,    
\end{equation}
and on the other side
\begin{equation}\label{eqproof3}
 i_{\Gamma_\alpha} \Lie_Z \eta^\alpha_L = \Lie_Z i_{\Gamma_\alpha} \eta^\alpha_L - i_{[Z,\Gamma_\alpha]} \eta^\alpha_L = - \Lie_Z E_L + i_{[\Gamma_\alpha,Z]} \eta^\alpha_L =  \ds\frac{\partial L}{\partial s^\alpha} g^\alpha + i_{[\Gamma_\alpha,Z]} \eta^\alpha_L \,.   
\end{equation}
Therefore, from \eqref{eqproof2} and \eqref{eqproof3} we deduce that
\begin{equation}\label{eqproof4}
  \Gamma_\alpha(g^\alpha) = \ds\frac{\partial L}{\partial s^\alpha} g^\alpha + i_{[\Gamma_\alpha,Z]} \eta^\alpha_L\,.
\end{equation}
Finally, substituting the above expression in equation \eqref{eqproof1} we obtain
$$\Gamma_\alpha(g^\alpha - i_Z \eta^\alpha_L) = \ds\frac{\partial L}{\partial s^\alpha} (g^\alpha - i_Z \eta^\alpha_L) ,$$
and by Lemma \ref{lem:12} the functions $F^\alpha= g^\alpha - i_Z \eta^\alpha_L$ give a dissipation law.
\end{proof}

\begin{remark}
When the functions $g^\alpha=0$, for every $\alpha=1,\ldots,k$, equations \eqref{Cartnsym} reduce to equations \eqref{last} of infinitesimal Lagrangian $k$-contact symmetries, and the above theorem becomes Corollary \ref{cafe}.
\end{remark}
To illustrate the above theorem we will give an example of a vector field satisfying the conditions  \eqref{Cartnsym}.
\begin{example}
Let us consider again the equations modelling a damped vibrating string given by the Lagrangian function \eqref{eq:Lagrangian-vibrating-string-dis}. The vector field $Z = \ds\frac{\partial}{\partial q} + \ds\frac{\partial}{\partial s^t} + g^x \ds\frac{\partial}{\partial s^x}$ satisfies
$$ \Lie_Z \eta^t_L = 0\,, \qquad \Lie_Z \eta^x_L = \d g^x\,, \qquad \Lie_Z E_L = - \ds\frac{\partial L}{\partial s^t} g^t - \ds\frac{\partial L}{\partial s^x} g^x\,, $$
where the function $g^t$ is constant, $g^t=c$, and $g^x$ is an arbitrary function.
         
By Theorem \ref{disipthc}, the corresponding dissipation law is the map given by $$F = \left(k - i_Z \eta^t_L,g^x - i_Z \eta^x_L\right) = ( \rho v_t, - \tau v_x) \,, $$
which coincides with the dissipation law given by the vector field $\ds\frac{\partial}{\partial q}$ in the Example \ref{examplekcontsym}.
\end{example}

\begin{remark}
When the Lagrangian $L$ does not depend on the variables $s^\alpha$, for $\alpha=1,\dotsc,k$, the function $L$ is defined on the bundle $\tkq$ and, if we think of $Z$ as vector field on $\oplus^k\T Q$, the conditions \eqref{Cartnsym} are now
$$
    \Lie_Z \d\eta^{\alpha}_L= 0\,,\qquad \Lie_ZE_L = 0\,,
$$
which is the definition of a Cartan symmetry for the $k$-symplectic system $(\oplus^k\T Q, \d \eta^{\alpha}_L,E_L)$, see \cite{rsv07}.
Moreover, the above Theorem \ref{disipthc} becomes the Noether's Theorem 3.13 in \cite{rsv07}.
\end{remark}

\subsection{Newtonoid vector fields}\label{subnewt}

In \cite{BBS} the set of Newtonoid
vector fields were introduced in the framework of the $k$-symplectic formulation of autonomous first-order field theories, extending the work of G. Marmo and N. Mukunda in \cite{marmo86} for the case $k = 1$. In this section, we introduce this kind of vector fields in the $k$-contact framework. 

In Proposition \ref{prop:csn} we prove that, for a regular Lagrangian function $L$, infinitesimal $k$-contact symmetries are Newtonoid vector fields for every corresponding {\sc sopde} ${\bf \Gamma} \in
\X^k_L(\oplus^k\T Q\times\R^k)$. Finally, we observe that a particular kind of Newtonoid vector fields that leave the Lagrangian function invariant are also infinitesimal $k$-contact symmetries and hence they provide dissipation laws.
\begin{definition}
    Consider a fixed {\sc sopde} ${\bf \Gamma}$. A vector field $X\in\mathfrak{X}(\oplus^k\T Q\times\R^k)$ is {\bf Newtonoid} with respect to ${\bf \Gamma}$ if $J^\alpha\left(\left[\Gamma_\alpha, X\right]\right)=0$. We denote by $\mathfrak{X}_\Gamma$ the set of Newtonoid vector fields associated to a \textsc{sopde} ${\bf\Gamma}$.
\end{definition}
Since
\begin{align}
    \left[\Gamma_\alpha, X\right] &= \left[v_\alpha^i \frac{\partial}{\partial q^i} + \Gamma^i_{\alpha\beta} \frac{\partial}{\partial v^i_{\beta}} + \Gamma_{\alpha}^\beta \frac{\partial}{\partial s^{\beta}}, X^i \frac{\partial}{\partial q^i} + X^i_{\gamma} \frac{\partial}{\partial v^i_{\gamma}} + X^{\gamma} \frac{\partial}{\partial s^{\gamma}} \right ] \nonumber\\  
&= \left( \Gamma_\alpha(X^i) -X^i_{\alpha}\right)
\frac{\partial}{\partial q^i} + \left( \Gamma_\alpha(X^i_{\beta})
  -X(\Gamma^i_{\alpha\beta})\right)\frac{\partial}{\partial v^i_\beta}+ \left( \Gamma_\alpha(X^{\beta})
  -X(\Gamma^\beta_\alpha)\right)\frac{\partial}{\partial s^\beta}\,, \label{ms0} 
\end{align}
we deduce that  $J^\alpha\left(\left[\Gamma_\alpha, X\right]\right)=0$ if, and only if, $\Gamma_\alpha(X^i) = X^i_\alpha$. Therefore, a Newtonoid vector field $X$ can be written locally as follows
\begin{equation}
X=X^i\frac{\partial}{\partial q^i} +
\Gamma_{\alpha}(X^i)\frac{\partial}{\partial v^i_{\alpha}}  + X^{\gamma} \frac{\partial}{\partial s^{\gamma}}\,.
 \label{xs2}
\end{equation}
From the local expressions \eqref{lifts-vectors1} and \eqref{xs2} it follows that the vector fields of the form
$$Z^C + X^\alpha \ds\frac{\partial}{\partial s^\alpha}\,,$$
where $Z^C\in
\X(\oplus^k\T Q\times\R^k)$ is the complete lift of a vector field $Z\in \X(Q)$ and $X^\alpha$ are arbitrary functions on $\oplus^k\T Q\times\R^k$, are
Newtonoid vector fields for an arbitrary {\sc sopde} ${\bf \Gamma}$.

In the next
proposition we see that the set of Newtonoid vector fields
contains also infinitesimal $k$-contact symmetries.

\begin{proposition} \label{prop:csn}
Let $X\in \X(\oplus^k\T Q\times\R^k)$ be an infinitesimal $k$-contact symmetry of a system described by a regular Lagrangian function $L:\oplus^k\T Q\times\R^k\to\R$. Then $X$ is a Newtonoid vector field for every ${\bf \Gamma} \in
\X^k_L(\oplus^k\T Q\times\R^k)$. 
\end{proposition}
\begin{proof}
Since ${\bf \Gamma}
\in\X^k_L(\oplus^k\T Q\times\R^k$)
and $L$ is regular, from Lemma \ref{lema} it follows that
${\bf \Gamma}$ is a {\sc sopde}. Moreover, ${\bf \Gamma}$ is a solution to the equations
\begin{equation}\label{eq:fields}
    i_{\Gamma_\alpha}\d\eta_L^\alpha=\d E_L+\displaystyle\derpar{L}{s^\alpha}  \eta_L^\alpha \,, \quad 
i_{\Gamma_\alpha}\eta_L^\alpha=-E_L \,.
\end{equation}
Since $\Lie_X \eta^\alpha_L = 0$ and $\Lie_X E_L = 0$, and using the second equation in \eqref{eq:fields} it follows that
$$ 0 = \Lie_X E_L = - \Lie_X i_{\Gamma_\alpha} \eta^\alpha_L = - i_{\Gamma_\alpha} \Lie_X \eta^\alpha_L - i_{[X,\Gamma_\alpha]} \eta^\alpha_L =  - i_{[X,\Gamma_\alpha]} \eta^\alpha_L\,. $$
 Now, if we apply $\Lie_X$ to
both sides of the first equation in \eqref{eq:fields} and use the commutation rules, we obtain
\begin{equation}\label{cafe1}
\Lie_X i_{\Gamma_\alpha}\d\eta_L^\alpha = \Lie_X \d E_L + \Lie_X \left(\parder{L}{s^\alpha} \eta_L^\alpha\right) =  \Lie_X\left(\parder{L}{s^\alpha}\right) \eta_L^\alpha + \parder{L}{s^\alpha}\Lie_X \eta_L^\alpha = \Lie_X\left(\parder{L}{s^\alpha}\right)\eta_L^\alpha
\end{equation}
and
\begin{equation}\label{cafe2}
\Lie_X i_{\Gamma_\alpha}\d\eta_L^\alpha =i_{\Gamma_\alpha}\Lie_X \d\eta^{\alpha}_L -
i_{[\Gamma_{\alpha}, X]}\d\eta^{\alpha}_L = -i_{[\Gamma_{\alpha}, X]}\d\eta^{\alpha}_L\,.
\end{equation}
Then, from \eqref{cafe1} and \eqref{cafe2} we have
\begin{equation}
     \label{ms1}
  i_{[\Gamma_{\alpha}, X]}\d\eta^{\alpha}_L= - 
  \Lie_X\left(\parder{L}{s^\alpha}\right) \eta_L^\alpha\,.
\end{equation}
We will prove now that equation \eqref{ms1} implies that
$J^{\alpha}([\Gamma_{\alpha}, X])=0$ and hence $X$ is a Newtonoid vector
field for ${\bf \Gamma}$.
Using formula \eqref{ms0}, we have
\begin{equation}\label{ms2}
    [\Gamma_{\alpha},X] = V_{\alpha}^i\parder{}{q^i} + V^i_{\alpha\beta} \parder{}{v^i_{\beta}} + V_{\alpha}^\beta\parder{}{s^\beta}\,,
\end{equation}
where
$$ V^i_{\alpha}=\Gamma_{\alpha}(X^i)-X^i_{\alpha}\,,\quad V^i_{\alpha\beta}=\Gamma_{\alpha}(X^i_{\beta})-
X(\Gamma^i_{\alpha\beta})\,, \quad V^\beta_{\alpha}=\Gamma_{\alpha}(X^\beta)-X(\Gamma^\beta_\alpha)\,. $$
Using formula \eqref{eq:contact-forms}, it follows that the differential of the contact one-forms $\eta^\alpha_L$ can be written as follows 
\begin{equation}\label{ms3}
    \d\eta^{\alpha}_L=a^{\alpha}_{ij} \d q^i\wedge \d q^j + g^{\alpha\beta}_{ij} \d q^i \wedge \d v^j_{\beta}+ h^\alpha_{i\beta} \d q^i \wedge\d s^\beta\,,
\end{equation} 
where
$$ a^{\alpha}_{ij}= \frac{1}{2}\left( \frac{\partial^2 L}{\partial q^j \partial v^i_{\alpha}} - \frac{\partial^2 L}{\partial q^i \partial v^j_{\alpha}} \right)\,, \quad g^{\alpha\beta}_{ij} = \frac{\partial^2 L}{\partial v^i_{\alpha} \partial v^j_{\beta}}\,, \quad h^\alpha_{i\beta} = \parderr{L}{s^\beta}{v^i_\alpha}\,. $$
If we replace now formulae \eqref{ms2} and \eqref{ms3} in equation
\eqref{ms1} we obtain
$$ \left(2a^{\alpha}_{ij} V^j_{\alpha} -
  g^{\alpha\beta}_{ij}V^j_{\alpha\beta} - h^\alpha_{i\beta} V_\alpha^\beta\right) \d q^i + g^{\alpha
  \beta}_{ij} V^i_{\alpha} \d v^j_{\beta}+ h^\alpha_{i\beta} V^i_\alpha   \d s^\beta=     - 
  \Lie_X\left(\displaystyle\derpar{L}{s^\alpha} \right) \left(\d s^\alpha - \ds\frac{\partial L}{\partial v^i_\alpha} \d q^i\right)\,, $$
which implies that $g^{\alpha \beta}_{ij} V^i_{\alpha}=0$. Using the fact that the Lagrangian $L$
is regular it follows that $g^{\alpha \beta}_{ij} $ has maximal rank
and hence $V^i_{\alpha}=\Gamma_{\alpha}(X^i)-X^i_{\alpha}=0$, which shows
that $X$ is a Newtonoid vector field for the $k$-vector fields ${\bf\Gamma}$.
\end{proof}

\begin{corollary} Let $L\in \Cinfty(\oplus^k\T Q\times\R^k)$ be a regular Lagrangian function.
\begin{enumerate}
\item Consider a Newtonoid vector field $X\in \X(\oplus^k\T Q\times\R^k)$ of the form
\begin{equation}
    X = Z^C + K^\alpha \frac{\partial}{\partial s^\alpha}\,,\quad\text{with } Z\in \X(Q)\,,
\end{equation}
and such that $X(L) = 0$, and $K^\alpha$ are constants, for $\alpha=1,\ldots,k$. Then $X$ is an infinitesimal $k$-contact symmetry.
\item The functions $F^\alpha= -i_X \eta^\alpha_L = Z^{V_\alpha}(L) - K^\alpha$ give a dissipation law. 
\end{enumerate}
\end{corollary}
\begin{proof}
    By a straightforward computation in local coordinates, it can be proved that $X$ is an infinitesimal $k$-contact symmetry, that is, it satisfies equations \eqref{last}. Then, by Corollary \ref{cafe} the functions $F^\alpha$  give a dissipation law for the  $k$-contact symmetry $X$.
\end{proof}

\section{Conclusions and further research}

In this paper we have presented several types of symmetries of non-conservative Lagrangian field theories using the $k$-contact framework. In particular, we have introduced natural symmetries (symmetries of the Lagrangian function), dynamical symmetries (preserving the solutions of the field equations) and $k$-contact symmetries (preserving the geometric structures). We have studied the relations among these symmetries and how to obtain dissipation laws from them. This theory has been illustrated with several physical examples.

The study of the symmetries of $k$-contact systems in the Lagrangian setting presented in this paper is another step towards a deeper study of the symmetries of non-conservative field theories. In future works, this should be complemented with the analysis of the symmetries and Newtonoid vector fields in the Hamiltonian counterpart of the $k$-contact formalism. In particular, it would be interesting to compare both sets of symmetries in the case of singular Lagrangians, when the Legendre map is not a diffeomorphism. In addition, it would be very interesting to study the symmetries and Newtonoid vector fields of these systems and relate them to the symmetries in the $k$-contact setting presented in this work.

Recently, a more general geometric framework for non-conservative field theories (generalizing the $k$-contact \cite{Gaset2021b} and $k$-cocontact \cite{Riv-2022} formulations) has been introduced: the so-called \emph{multicontact formalism} \cite{LGMRR-2022}. We propose to analyze the relation between the $k$-contact and the $k$-cocontact formulations with the multicontact setting, both in the Lagrangian and Hamiltonian formalisms, following the ideas in \cite{relations}.

\section{Acknowledgments}

X. Rivas, M. Salgado and S. Souto acknowledge financial support of the Ministerio de Ciencia, Innovaci\'on y Universidades (Spain), projects PGC2018-098265-B-C33 and D2021-125515NB-21.

X. Rivas acknowledges financial support of the Novee Idee 2B-POB II project PSP: 501-D111-20-2004310 funded by the ``Inicjatywa Doskonałości - Uczelnia Badawcza'' (IDUB) program.



\bibliographystyle{abbrv}
\bibliography{references.bib}

\end{document}